\documentclass[11pt,a4paper]{article}

\usepackage{geometry}
\usepackage{jheppub}
\usepackage[utf8]{inputenc}
\usepackage{booktabs} 
\usepackage{array} 
\usepackage{paralist} 
\usepackage{verbatim} 
\usepackage{mathtools}
\usepackage{amsmath,amssymb,amsthm}
\usepackage{mathrsfs}
\usepackage{microtype}
\usepackage{cleveref}
\usepackage{slashed}
\usepackage{bbold} 
\usepackage[justification=centering]{caption}
\usepackage{subcaption}
\usepackage{tikz}

\newcommand{\mc}[1]{\mathcal{#1}}
\newcommand{\ii}{\text{i}}

\newcommand{\dd}{\textrm{d}}
\newcommand{\db}{\bar{\partial}}

\newcommand{\ul}[1]{\underline{#1}}

\newcommand{\bsym}[1]{\boldsymbol{#1}}

\swapnumbers
\newtheorem{theorem}{Theorem}[subsection]
\newtheorem*{theorem*}{Theorem}

\newtheorem*{prop*}{Proposition}

\newtheorem*{conj*}{Conjecture}

\newtheorem*{prin*}{Principle}

\newtheorem{lemma}[theorem]{Lemma}
\theoremstyle{definition}

\newtheorem*{defn*}{Definition}
\newtheorem{ex}[theorem]{Example}
\newtheorem{remark}[theorem]{Remark}

\newtheorem*{warning*}{Warning}

\title{Vortex counting and the quantum Hall effect}
\author{Edward Walton}
\affiliation{{Department of Applied Mathematics and Theoretical Physics, University of Cambridge, Wilberforce Road, Cambridge, CB3 0DZ, UK}}
\date{}
\emailAdd{e.walton@damtp.cam.ac.uk}
\abstract{
We provide evidence for conjectural dualities between nonrelativistic Chern--Simons-matter theories and theories of (fractional, nonAbelian) quantum Hall fluids in \(2+1\)~dimensions. At low temperatures, the dynamics of nonrelativistic Chern--Simons-matter theories can be described in terms of a nonrelativistic quantum mechanics of vortices. At critical coupling, this may be solved by geometric quantisation of the vortex moduli space. Using localisation techniques, we compute the Euler characteristic \({\chi}(\mc{L}^\lambda)\) of an arbitrary power \(\lambda\) of a quantum line bundle \(\mc{L}\) on the moduli space of vortices in \(U(N_c)\) gauge theory with \(N_f\) fundamental scalar flavours on an arbitrary closed Riemann surface. We conjecture that this is equal to the dimension of the Hilbert space of vortex states when the area of the metric on the spatial surface is sufficiently large. We find that the vortices in theories with \(N_c = N_f = \lambda\) behave as fermions in the lowest nonAbelian Landau level, with strikingly simple quantum degeneracy. More generally, we find evidence that the quantum vortices may be regarded as composite objects, made of dual anyons. We comment on potential links between the dualities and three-dimensional mirror symmetry. We also compute the expected degeneracy of local Abelian vortices on the \(\Omega\)-deformed sphere, finding it to be a \(q\)-analog of the undeformed case.
}

\begin{document}

\maketitle

\section{Introduction and summary}

In this paper we investigate the low-temperature behaviour of nonrelativistic gauge theory in \(2+1\) dimensions.
We focus particularly on \(U(N_c)\) Chern--Simons gauge theories with \(N_f\) fundamental scalar fields compactified on Riemann surfaces. In the Higgs phase, these theories are dynamically rich, supporting vortex solutions. Indeed, the particular theories we study have an effective low-temperature description in terms of Hamiltonian dynamics on moduli spaces of (nonAbelian) vortices. 

\emph{A priori}, studying the dynamics of these theories, even at low temperature, looks like a difficult task. Vortices are complicated objects, solving systems of coupled, nonlinear differential equations to which analytic solutions are known only in extremely special cases. As soon as \(N_c > 1\), the geometry of the moduli space is - in the best cases - only partially understood, making it hard to study quantum mechanics on it directly. With all that said, there are two remarkable strands of knowledge which allow us to get to grips with these theories, as follows.
\begin{itemize}
	\item \emph{Coulomb branch localisation}. Also known as Jeffrey--Kirwan--Witten localisation and originally due to Witten \cite{witten2DGT} and Jeffrey--Kirwan \cite{jeffreyLNA}, this technique allows one to express integrals over symplectic quotient spaces in terms of certain residues. It has led to enormous progress in the analysis of supersymmetric field theories (see \cite{pestunLT}, for instance). The vortex moduli space is a symplectic quotient (of an infinite-dimensional space by an infinite-dimensional group), so this is available to us (at least formally). This technique has been used in \cite{miyakeVMS,ohtaHCB} to compute the volumes of vortex moduli spaces.
	\item \emph{Duality}. It is expected that vortices, these magnetically charged particles in the background of an electrically charged condensate, admit a dual description as electrically charged particles in the background of a magnetic field.  This idea has found itself to be particularly important in understanding the physics of the quantum Hall effect: it is expected that fluids of vortices in bosonic Chern--Simons-matter theories with gauge group \(U(N_c)\) at Chern--Simons level \(\lambda\) describe (nonAbelian) quantum Hall fluids at filling fraction \(N_c/\lambda\) \cite{sternQHF,tongQHF,radicevicNB,doreyVM,kimuraVD}.
\end{itemize}
We will weave these strands together, using the first to provide quantitative evidence for the second and using the second to guide our interpretation of the results of the first.

Our main result is the computation of the Hilbert polynomials \(\chi(\mc{L}^\lambda)\) of quantum line bundles \(\mc{L} \to \mc{M}\) on moduli spaces of nonAbelian vortices for arbitrary \(N_c\) and \(N_f\) on an arbitrary compact Riemannian surface \(\Sigma\). This corresponds to computing the index (that is, the `expected' dimension of the Hilbert space) of the effective topological quantum mechanics of the \(2+1\) dimensional theories on \(\Sigma \times S^1\) or \(\Sigma \times \mathbb{R}\) at arbitrary Chern--Simons level, which is the power \(\lambda\). We expect, but do not prove, that this index exactly computes the Hilbert space dimension when the area of \(\Sigma\) is sufficiently large. 

In the Abelian case, where the moduli spaces are relatively easy to describe, we compute the index directly, using the Hirzebruch--Riemann--Roch theorem to express \(\chi(\mc{L}^\lambda)\) as an integral over the moduli space and doing this integral directly (this is sometimes called `Higgs branch localisation'). In the nonAbelian case, where the moduli spaces are not sufficiently well-understood to do this, we use Coulomb branch localisation techniques to carry out the computation. A key step in the calculation is the identification of the correct residue prescription when negative powers of Vandermonde-type determinants appear (which is whenever the genus of \(\Sigma\) is one or higher). The form that the result takes is generally complicated, but in special cases it simplifies dramatically.

Our approach runs parallel to matrix model approaches to understanding quantum mechanics on the vortex moduli space \cite{polyQH,hellermanQH,hananyVIB,tongQHSS, doreyVM}. The advantage of our approach is that it is rather general (in particular, we can work on surfaces of arbitrary genus) and, in a sense, geometrically clearer. On the other hand, we are restricted to the computation of rather `soft' observables (like the Euler characteristic) and may have to work hard for concrete insight.

The simplest, and perhaps most striking, upshot of our study is strong evidence for the notion that nonrelativistic vortices in \(U(N)\) gauge theories with \(N\) fundamental Higgs fields (so-called \emph{local} vortices) at level \(\lambda = N\) `are' fermions in a background flux quantum mechanically, at least at low temperature. This can be expressed as a low-temperature duality
\begin{equation}
\label{eq:duality1}
Z_{\text{SCS},N} \,\, \leftrightarrow \,\, Z_{\text{Fermi},N}
\end{equation}
where \(Z_{\text{SCS},N}\) is a critically-coupled nonrelativistic Chern--Simons-matter theory (a Schr\"odinger--Chern--Simons theory, hence SCS) with gauge group \(U(N)\), Chern--Simons level \(\lambda = N\) and \(N\) fundamental flavours, and \(Z_{\text{Fermi},N}\) is a theory of \(N\) fermion flavours critically coupled to a particular background gauge field for the \(U(N)\) flavour symmetry. Vortices in the theory on the left are mapped to elementary fermionic excitations on the right. The matching of various objects across the duality is summarised in \autoref{tab:duality}.

\begin{table}[h!]

\begin{center}
\begin{tabular}{c c}
 \(Z_\text{SCS}\) & \(Z_\text{Fermi}\)
\rule[-1.5ex]{0pt}{0pt} \\
\hline  \rule{0pt}{1.0\normalbaselineskip}
\text{vortices} & \text{fermionic particles} \\
\text{vortex charges} & \text{fermion flavours} \\
\text{monopole operators} & \text{particle creation operators} \\
\text{vortex-hole symmetry} & \text{particle-hole symmetry}\\
\\
\text{\((U(1)_\text{top} \times SU(N)_\text{flavour})/\mathbb{Z}_N\) global symmetry} & \text{\(U(N)\) global symmetry}\\
\\
\text{Fayet--Iliopoulos parameter} & \text{background \(U(N)\) gauge field} \\
\text{Bohr--Sommerfeld quantisation condition} & \text{flux quantisation condition} \\
\\
\text{saturated Bradlow bound} & \text{saturated Landau level}\\
\\
\end{tabular}
\caption{A summary of the local Fermi-vortex duality.}
\label{tab:duality}
\end{center}
\end{table}

The duality \eqref{eq:duality1} was more-or-less proved in the Abelian case of \(N = 1\) in \cite{erikssonKQ}. In particular, it was shown in \cite{erikssonKQ} that the low-temperature topological quantum mechanics on both sides of the duality were equivalent: an isomorphism of Hilbert spaces was given. At the level of indices, the same result was found in \cite{deyGQ}.

Both \cite{erikssonKQ} and \cite{deyGQ} use knowledge of the geometry of the moduli space of Abelian vortices in their calculations. As we resort to the use of rather `soft' techniques, we demonstrate a weaker statement for general \(N\), finding an equality of the `expected' dimension of the Hilbert space (that is, the relevant index) on both sides of the duality\footnote{We expect that the index exactly computes the dimension of the Hilbert space when the area of \(\Sigma\) is sufficiently large. This expectation is based on the general lore of positivity for vector bundles and is proved in the Abelian case in \cite{erikssonKQ}.}. In particular, on both sides of the duality we find the index (at fixed particle number) 
\begin{equation}
\label{eq:duality1-1}	
{N \mc{A} \choose k}
\end{equation}
where \(N \mc{A} = N \frac{e^2 \tau \text{vol}(\Sigma)}{4\pi} \) is the area of the two-dimensional surface \(\Sigma\) in units of vortex size (it is related linearly to the Fayet--Iliopoulos parameter \(\tau\) of the vortex theory, and to the gauge coupling constant \(e^2\)) and \(k\) is the number of particles or vortices. We see that the inverse of the vortex size, which is roughly the charge density of the bosonic condensate, sets the strength of the effective magnetic field in the Fermi theory. The parameter \(N\mc{A}\) must be an integer: in the vortex theory this is necessary for the quantum theory to exist, in the Fermi theory this is a flux quantisation condition. 

On the Fermi side, the computation of the index is a simple calculation: these are fermions in the lowest (nonAbelian) Landau level at filling fraction \(\nu = k/N\mc{A}\). On the vortex side it is tricky: even after one has carried out the localisation calculation, it involves apparently miraculous cancellations in a polynomial in \(N\mc{A}\) of generic degree \(Nk\) (there is also a simple, but highly heuristic, argument for the result just from symmetry arguments, which we give in \autoref{subsec:heuristic}). 

The result \eqref{eq:duality1-1} does not depend on the topology or local geometry of the surface \(\Sigma\), which may be thought of as reflecting the local nature of vortices when \(N_f = N_c\). The dependence of the result on the choice of two-dimensional scale \(\mc{A}\) is simple. Note also that \eqref{eq:duality1-1} captures the so-called Bradlow bound \cite{bradlowVLB}, which is the selection rule  \(0 \leq k \leq N\mc{A}\) for vortices on compact surfaces.

The fact that our evidence takes the form of an equality of indices may lead one to suspect that the duality only holds as a result of some hidden supersymmetry. However, the evidence of \cite{erikssonKQ} in the Abelian case suggests that this is not the case. In that case there is a genuine isomorphism of non-supersymmetric Hilbert spaces. Our methods are not powerful enough to show this in the nonAbelian case, but it seems reasonable to conjecture that it is true; this conjecture then leads to conjectures on the geometry of the moduli space of nonAbelian vortices.

Our general result, given as \eqref{eq:generalvortexcount}, contains much more than \eqref{eq:duality1-1}. We can independently vary \(\lambda\), which is conjectured to change the filling fraction of the dual quantum Hall fluid, \(N_f\), and  \(N_c\). In general, these changes lead to significant complications in the form of the index, involving genus-dependent contributions. This reflects just how special local vortices at the right level are. 

So as to not bombard the reader with unnecessarily complicated formulae at this early stage, we will illustrate the ideas in the Abelian case. Setting \(N_c = 1\) but allowing for a general number \(N_f\) of charged scalars and general level \(\lambda\), we find the index
\begin{equation}
\label{eq:duality1-2}
\sum_{j=0}^g {g \choose j} \lambda^j N_f^{g-j} {\lambda( \mc{A} - k) + N_f k + (N_f-1)(1-g) - g \choose N_f k + (N_f -1)(1-g) - j }
\end{equation}
at vortex number \(k\), where \(g\) is the genus of \(\Sigma\). We derive this result rigorously: it is given as \autoref{theorem:semilocalab}.

We extract some lessons from this, as follows.
\begin{itemize}
	\item \emph{Composite vortices and fractional statistics}. The way that the vortex number \(k\) enters \eqref{eq:duality1-2} indicates that vortices can be given a `composite' interpretation. If we squint and ignore topology-dependent effects for a moment, comparison with counting formulae for particles with fractional statistics (see \cite{wuSD}, for example) suggests that we should think of a single vortex in this Abelian theory as a bound state of particles with exchange phase \(\exp(\ii \pi N_f/\lambda)\) (and, dually, vortex holes are bound states of particles with exchange phase \(\exp(\ii \pi \lambda/N_f)\)). There are some subtle topological issues in play, but: \emph{Quantum vortices are made of anyons!} This idea is not new: in \cite{tongQHSS} it was pointed out that (local) Abelian vortices can admit quasihole-like excitations with fractional spin and charge, which were evaluated using a Berry phase argument.
	 
The composite vortices can be thought of as experiencing the `effective magnetic flux'
\[
\lambda\mc{A} - (\lambda - N_f)k + \cdots = N_f \mc{A} - (N_f - \lambda)(\mc{A} - k) + \cdots 
\]
where \(\cdots\) denotes topological contributions. This is highly reminiscent of the effective magnetic flux experienced by the composite fermions of Jain \cite{jainCFA,jainCF}, with \(\mc{A}\) identified with the total applied magnetic flux (as in the duality above). Indeed, modulo topology, at \(N_f = 1\) and \(\lambda\) odd, vortices should be dual to composite fermions, electrons bound to \(\lambda-1\) flux quanta.   

We note further that the topological contributions that we have suppressed above are related in the dual theory to the geometrical phases of  \cite{wenSS} that arise when coupling the quantum Hall fluid to the curvature of space (see also \cite{haldaneGD} for an approach to understanding quantum Hall fluids where similar phases arise, albeit for slightly different reasons).

	\item \emph{Vortex-hole (a)symmetry and bosonisation}. An interesting feature of the local result \eqref{eq:duality1-1}	is that it demonstrates an exact symmetry between the theory of the vortices in the Higgs phase, which have number \(k\), and the theory of vortex holes in a vortex fluid (which is a Coulomb phase), which have number \(N_c\mc{A} - k\), via the symmetry \({N\mc{A} \choose k} = {N\mc{A} \choose N\mc{A} - k}\). 
	
	The index \eqref{eq:duality1-2} illustrates the fact that when \(N_c \neq N_f\) this symmetry is generally broken by topology-dependent effects. In the general case, this breaking is characterised by the quantity\footnote{This quantity is, not by coincidence, reminiscent of that characterising the ghost anomaly in the theory of two-dimensional sigma models into \(\text{Gr}(N_c,N_f)\) with \(\mc{N}= (2,2)\)  supersymmetry.} \(N_c(N_f-N_c)(1-g)\). The symmetry is partially restored at \(g=1\): in this case, \eqref{eq:duality1-2} is unchanged if one interchanges
	\begin{equation}
	\label{eq:vhdual}
	k \leftrightarrow \mc{A}-k \text{ and }\lambda \leftrightarrow N_f \text{.}
	\end{equation}
	
That one must interchange \(\lambda\) with \(N_f\) is perhaps surprising and a little mysterious. A potential explanation, which we do not develop in any detail here, comes from the consideration of more general  three-dimensional theories. The Chern--Simons level may be viewed as being induced by integrating out massive fermionic fields in an auxiliary ultraviolet theory. Then the genus one vortex-hole duality of \eqref{eq:vhdual} might be viewed as a shadow of (a nonrelativistic version of) the `master' \emph{bosonisation} duality of \cite{jensenMB} (which builds, in particular, on \cite{aharonyBMD}; see also \cite{radicevicNB} for comments on bosonisation in nonrelativistic theories). This duality interchanges fermions and bosons, so by applying the duality and then integrating out the fermions on both sides, one recovers the right kind of picture. 
	
	\item \emph{Topological degeneracy}. Consider the vortex fluid at maximal density, corresponding to \(k = \mc{A}\). Then the index \eqref{eq:duality1-2} tells us that the vortex fluid has expected quantum degeneracy \(\lambda^g\). 
	
	This is not surprising. When \(k = \mc{A}\), the vortex centres spread over the whole surface, the Higgs fields vanish everywhere, and the theory is in a Coulomb phase, with a massless photon. The theory is then simply a \(U(1)\) Chern--Simons theory at level \(\lambda\). This theory has classical degeneracy characterised by the moduli space of flat \(U(1)\) connections. It is well-known, and not too hard to show, that the corresponding quantum theory has degeneracy \(\lambda^g\).
	
\end{itemize}

The above ideas extend to the nonAbelian case, but now there are some wrinkles. The outcome of the computation becomes significantly more complicated away from the special cases of \(N_c = N_f = \lambda\) or \(N_c =1\). We expect that the general result leads to understanding of nonAbelian fractional quantum Hall fluids. 

The rest of this note is structured as follows. In \autoref{sec:csm} we outline the mathematical background behind nonrelativistic Chern--Simons-matter theories and describe how they can be effectively described in terms of quantum mechanics on the vortex moduli space. We also describe the process of geometric quantisation and some of the subtleties that arise. In \cref{sec:avm,sec:navm,sec:simp} we carry out the computation of the Hilbert polynomial of the quantum line bundle on the vortex moduli space. This computation has three parts:
\begin{itemize}
	\item Rigorously quantise the moduli space of Abelian vortices. We do this by direct integration over the moduli space, which we describe as a projective bundle over the Jacobian variety of \(\Sigma\) for sufficiently large \(k\). The result is given in \autoref{theorem:semilocalab}. When \(N_c = N_f =1\) and \(\Sigma \cong S^2\), we go further, quantising the theory in the presence of a `harmonic trap'.  
	
	\item Use Coulomb branch localisation to reduce the computation in the nonAbelian case to an integral over a Cartan subalgebra of \(\mathfrak{u}(N_c)\) and, correspondingly, to a sum of `Abelian' contributions, which we now understand. The result of this is the general formula \eqref{eq:generalvortexcount} for the Hilbert polynomials of quantum line bundles on moduli spaces of vortices in the theories we study. 
	
	\item Simplify the output in special cases. This is `elementary' but tricky, requiring us to prove some combinatorial identities. Our main weapon in these proofs is simply the comparison of generating functions. The result in the case \(N_c = N_f = \lambda\) is given in \eqref{eq:localcount}.
\end{itemize}
We go on to interpret the results in the context of duality in \autoref{sec:localduality} and \autoref{sec:slduality}.

\section{Chern--Simons-matter theories and vortices}
\label{sec:csm}

\subsection{K\"ahler vortices}
\label{subsec:kahlervortices}

Let \((Y, \omega_Y)\) be a K\"ahler manifold carrying the isometric, Hamiltonian action of a compact Lie group \(G\). Write 
\[
\mu : Y \to \mathfrak{g}^\vee
\]
for the K\"ahler moment map. Here \(\mathfrak{g}\) is the Lie algebra associated to \(G\). We will identify \(\mathfrak{g}\) with its dual using the Killing form. 

\begin{ex}
The example to which we specialise shortly is the case where \(Y\) consists of \(N_f\) copies of the fundamental representation of \(G = U(N_c)\). The moment map is then
\[
\mu((z_1, \cdots, z_{N_f})) = \frac{\ii}{2} z_i z_i^\dagger 
\]	
where \(z_i \in \mathbb{C}^{N_c}\) for \(i = 1, \cdots, N_f\) and summation is implied.
\end{ex}

Let \(\Sigma\) be a Riemann surface, representing physical space, with a volume form \(\omega_\Sigma\). Let \(P \to \Sigma\) be a principal \(G\)-bundle. We can then form the associated \(Y\)-bundle
\[
p : \mc{Y} \coloneqq P \times_G Y \to \Sigma. 
\]
We consider a theory of a connection \(A\) on \(P\) and a section \(\phi\) of \(\mc{Y}\). The connection \(A\) on \(P\) induces a connection on \(\mc{Y}\), which is a splitting of the Atiyah exact sequence
\begin{equation}
\label{eq:atiyah}
T_V \mc{Y} \to T\mc{Y} \xrightarrow{\dd p} p^*T\Sigma \text{,}
\end{equation}
where \(T_V \mc{Y}\) is the bundle of vertical vector fields (defined as those in the kernel of \(\dd p\)). Indeed, the connection \(A\) induces a map 
\[
v_A : T\mc{Y} \to T_V\mc{Y}
\]
which splits the exact sequence \eqref{eq:atiyah}. The \emph{covariant derivative} of a section \(\phi\) of \(\mc{Y}\) with respect to the connection \(A\) is then defined by
\[
\dd_A \phi =  v_A ( \dd \phi ) \in \Omega^1(\Sigma, T_V \mc{Y}) \text{.}
\]

The symplectic form \(\omega_Y\) on \(Y\) induces a vertical 2-form on the fibres of \(\mc{Y}\), but not a true 2-form. For that, we need to use the connection to tell us how to feed non-vertical vectors into the vertical 2-form (this is essentially minimal coupling to a gauge field). This gives a true 2-form, which we call \(\omega_\mc{Y}\).

The natural energy density functional on the space of pairs \((A, \phi)\) is the Yang--Mills--Higgs energy density
\[
\mc{E}_\text{YMH} = \frac{1}{e^2} |F(A)|^2 + |\dd_A\phi|^2 + e^2(\xi - \mu(\phi))^2 \text{,}
\]
where \(\xi\) is an element of the (dual of the) centre of \(\mathfrak{g}\) (known in the context of supersymmetric field theory as the \emph{Fayet-Iliopoulos parameter}), and \(e^2\) is the gauge coupling constant, a scale factor on the dual Killing form on \(\mathfrak{g}^\vee\) (we call the coupling constant \(e^2\) even in the nonAbelian case rather than the usual \(g^2\) to avoid confusion with the genus \(g\) of \(\Sigma\)). As shown in \cite{cieliebakSV}, this functional admits the following \emph{Bogomolny rearrangement}
\begin{equation}
\label{eq:vortexbog}
\mc{E}_\text{YMH} = \frac{1}{e^2} | *F(A)- e^2\left(\xi - \mu(\phi) \right) |^2 + |\db_A \phi|^2 + * (\phi, A)^* [\omega_Y^G](\xi) \text{,}
\end{equation}
where the Dolbeault operator \(\db_A\) is defined by
\[
\db_A \phi = \frac{1}{2} \left( \dd_A \phi + J_Y \circ \dd_A\phi \circ j_\Sigma\right) \text{,}
\]
and where \(\omega^G_Y\) is the \(G\)-equivariant symplectic form on \(Y\) (see \cite{atiyahMM}).

The \emph{vortex equations} ask that
\begin{align}
F(A) &= e^2 \left(\xi - \mu(\phi) \right) \omega_\Sigma \\
\db_A \phi &= 0 \text{.}
\end{align}
Solutions to these equations are absolute minimisers of \(\mc{E}_\text{YMH}\) within their topological class.

\subsection{Vortex moduli}

Write \(\mc{C}\) for the infinite-dimensional space of (reasonable) pairs \((A, \phi)\). This is an infinite-dimensional K\"ahler manifold, with K\"ahler form\footnote{We have inserted an extra factor of \(e^2\) into the definition of the K\"ahler form so as to simplify notation later, in line with our convention that Chern--Simons levels are pure integers.}
\[
\omega_\mc{C} ( (\dot{A}_1, \dot{\phi}_1), (\dot{A}_2, \dot{\phi}_2) ) = e^2 \int_\Sigma \left(\frac{1}{e^2} \text{tr}(\dot{A}_1 \wedge \dot{A}_2) + \omega_\mc{Y} ( \dot{\phi}_1, \dot{\phi}_2)  \omega_\Sigma  \right)
\]
where \(( \dot{A}_i, \dot{\phi}_i) \in \Omega^1(\Sigma, \text{ad}_P) \times\Gamma(\phi^*T_V\mc{Y})\), \(i=1,2\) are  tangent vectors to \(\mc{C}\) at a point \((A,\phi)\). The space \(\mc{C}\) has a natural complex structure and the corresponding Riemannian metric is the natural kinetic energy functional for the relativistic \((2+1)\)-dimensional gauge theory associated to this data. 

Let \(\mc{C}_0\) be the space of pairs \((A, \phi) \in \mc{C}\) solving the equation
\[
\db_A \phi = 0 \text{.}
\]
This is (formally) a symplectic space, with symplectic form \(\omega_{\mc{C}_0}\).
The group \(\mc{G} \coloneqq \Gamma(\text{Ad}_P)\) of gauge transformations acts on \(\mc{C}\) and on \(\mc{C}_0\) in the usual way, preserving the symplectic forms \(\omega_\mc{C}\) and \(\omega_{\mc{C}_0}\). One can show \cite{cieliebakSV} (see also \cite{atiyahYM}) that the moment map for this action is
\[
\nu(A,\phi) = * F(A)  + e^2 \mu(\phi)
\]
which takes values in the dual of the Lie algebra of \(\mc{G}\). This tells us that the vortex moduli space is the K\"ahler quotient
\[
\mc{M} = \mc{C}_0 //_\xi \mc{G} \text{.}
\]
We have been very sloppy: the functional analysis of this was worked out in the Abelian case in \cite{garciapradaDEP}, and it works out as advertised. One upshot of this result is that the vortex moduli space is a K\"ahler manifold (at least for generic \(\xi\)).

What does the moduli space look like? In general it is hard to say, but a great deal is known in special cases. In particular, if \(G = U(1)\) and \(Y\) is the fundamental representation of \(U(1)\), then the moduli space of charge \(k\) vortex solutions on \(\Sigma\) is well-known \cite{jaffeVM,noguchiYMH,bradlowVLB,garciapradaDEP} to be the symmetric product
\[
\text{Sym}^k(\Sigma) \coloneqq \Sigma^{\times k}/S_k
\]
where the symmetric group \(S_k\) acts by permuting the various factors in \(\Sigma^{\times k}\). This tells us that a charge \(k\) vortex solution in this case is uniquely specified by an unordered list of \(k\) points in \(\Sigma\), which we interpret as the positions of the \(k\) vortices.

There have been many more investigations into the moduli spaces of more general vortices from both a physical and mathematical perspective, including \cite{baptistaNAV,biswasMV,etoMS,etoSHP,etoCNA,etoOTM,hananyVIB}.

\subsection{Schr\"odinger--Chern--Simons theories}

The Chern--Simons matter theories which play a central role in our story are the `natural' nonrelativistic \((2+1)\)-dimensional extensions of the theory of two-dimensional Yang--Mills--Higgs theory we considered above. As we will show, these theories admit an (exact) description as gauged Hamiltonian mechanics on the infinite-dimensional configuration space \(\mc{C}\) and an effective low-temperature description in terms of trivial Hamiltonian mechanics on the vortex moduli space \(\mc{M}\). The particular theories we consider are generalisations of the Abelian theory considered by Manton in \cite{mantonFOV} (see also \cite{mantonCL,romaoQCS,romaoSSD,tongQHSS} for further considerations of Manton's model relevant to our study). We will give the general formulation to illustrate the geometry of the theories, but will later specialise to familiar linear examples, where the description simplifies. 

The theory is specified by the data of a Lie group \(G\), a Hamiltonian \(G\)-space \(Y\), as above, and a Chern--Simons level in \(H^4(BG)\), which we represent as \(\lambda\). The action of the theory on \(\Sigma \times S^1\) takes the form
\begin{equation}
\label{eq:lagcsm}
\begin{gathered}
S = \frac{\lambda}{2\pi} \int_{\Sigma \times D^2} \left( \text{tr}\left( F(\tilde{A}) \wedge F(\tilde{A}) \right) + e^2 \tilde{\phi}^*\omega_\mc{Y} \wedge \omega_\Sigma + e^2 (\mu(\tilde{\phi}) - \xi, F(\tilde{A}) )\wedge \omega_\Sigma \right) \\
 \quad \quad \quad\quad\quad\quad\quad - \frac{\beta}{2\pi} \int_{\Sigma \times S^1} \mc{E}_\text{YMH} \dd t \wedge \omega_\Sigma 
\end{gathered}
\end{equation}
There are a few things to explain here. First, \(\tilde{A}\) and \(\tilde{\phi}\) are extensions of \(A\) and \(\phi\) respectively from \(\Sigma \times S^1\) to \(\Sigma \times D^2\). In general, such extensions may not exist, so in defining the theory we should restrict to loops \((A|_\Sigma,\phi) : S^1 \to \mc{C}\) which are simply connected. Provided an extension exists, the Lagrangian depends only on the homotopy class of the extension provided the parameter \(\lambda\) is an integer (or more accurately, that the class in \(H^4(BG)\) induced by \(\lambda\) is integral).\footnote{There is a quirk of conventions here. A Chern--Simons level is sometimes thought of as a dimensionful quantity taking values in \(\hbar e^{-2} \mathbb{Z}\), rather than a dimensionless quantity taking values in \(\mathbb{Z}\). We will prefer to take our Chern--Simons levels to be pure integers. We will usually work in units with \(\hbar = 1\), but we have to be careful with factors of \(e^2\).}

The coordinate \(t \in [0, 2\pi)\) is a periodic coordinate on the circle factor. The constant \(\beta\) is a scale factor on the length of the circle. Notice that it appears only in front of the second integral, as this is the only place where a `bare' \(\dd t\) sits. 

For reasons that we'll describe shortly, one can (up to a ground state energy) replace the critically-coupled two-dimensional energy density \(\mc{E}_\text{YMH}\), which is the Hamiltonian of the theory of \eqref{eq:lagcsm}, with \(|\db_A \phi|^2\), which looks a little less `fine-tuned'.

To connect \eqref{eq:lagcsm} with more familiar things, we can write it purely in \((2+1)\)-dimensional notation. To do this, it is useful to specialise to the case of \(G = U(N_c)\) and \(Y = \mathbb{C}^{N_cN_f}\). Then
\begin{align*}
\tilde{\phi}^* \omega_\mc{Y} - (\mu(\tilde{\phi}),F(A)) &= \ii \, \dd_{\tilde{A}} \tilde{\phi}_i^\dagger \wedge \dd_{\tilde{A}} {\tilde{\phi}}_i - \ii \phi^\dagger_i F(A) \phi_i \\
	&= \frac{\ii}{2} \dd \left( \tilde{\phi}^\dagger_i \dd_{\tilde{A}} \tilde{\phi}_i - ( \dd_{\tilde{A}} \tilde{\phi} ^\dagger_i) \tilde{\phi}_i \right)
\end{align*}
where \(i = 1, \cdots, N_f\) is the flavour index and summation is implied. This allows us to use Stokes' theorem, leaving us with an integral over the boundary of \(\Sigma \times D^2\), which is \(\Sigma \times S^1\). This is
\begin{equation}
\label{eq:lagcsm2}
\begin{gathered}
S = \frac{\lambda}{2\pi} \int_{\Sigma \times S^1} \left( *\text{CS}(A) - e^2 (\xi, A_t)  + \frac{\ii e^2}{2} \left( \phi_i^\dagger D_t \phi_i - (D_t \phi_i)^\dagger \phi_i \right) + \frac{\beta}{\lambda} \mc{E}_\text{YMH} \right) \dd t \wedge \omega_\Sigma
\end{gathered}
\end{equation}
where \(D_t\) is the time component of the covariant derivative and \(\text{CS}(A)\) is the Chern--Simons 3-form,
\[
\text{CS}(A) \coloneqq \text{tr}\left( A \wedge \dd A + \frac{1}{3} A \wedge [A,A]\right)\text{,}
\]
which has the defining property that it is a primitive of \(\text{tr}(F(\tilde{A}) \wedge F(\tilde{A}))\).

In the form \eqref{eq:lagcsm2}, the action looks quite familiar. It is the obvious generalisation of the model considered by Manton in \cite{mantonFOV} in the Abelian case, at least for certain choices of coupling. The nonAbelian version of this theory has previously been considered in \cite{radicevicNB,turnerBPS}, for example. The reason that we initially wrote \eqref{eq:lagcsm} rather than \eqref{eq:lagcsm2} is that the latter is not generally well-defined (or rather, it is well-defined only when one defines it as \eqref{eq:lagcsm}). 

\subsection{Classical vortex mechanics}

The action \eqref{eq:lagcsm} can be reinterpreted as an action for gauged Hamiltonian dynamics on the infinite-dimensional configuration space \(\mc{C}\). The gauge group is the group \(\mc{G} = \Gamma(\Sigma, \text{Ad}_P)\) of \(G\)-valued gauge transformations on \(\Sigma\).

In \((0+1)\)-dimensional language, the field content of the theory consists of a path in the space of configurations on \(\Sigma\), which we write as 
\[
(A_\Sigma, \phi) : S^1 \to \mc{C}
\]
where \(A_\Sigma\) is the two-dimensional gauge field and \(\phi\) is the scalar field, as before, and a gauge potential \(a\) for the infinite-dimensional gauge group \(\mc{G}\). The gauge potential \(a\) is the \(t\)-component of the original gauge potential \(A\) of the \((2+1)\)-dimensional theory. Indeed, \(A = a + A_\Sigma\).

More precisely, we take a \(\mc{G}\)-bundle \(\mc{P} \to S^1\) with connection \(a\).
There is an associated \(\mc{C}\)-bundle given by
\[
\ul{\mc{C}} \coloneqq \mc{P} \times_\mc{G} \mc{C} \to S^1 \text{,}
\]
and we take a section \((A_\Sigma, \phi)\) of this bundle. The space \(\mc{C}\) has a symplectic form \(\omega_\mc{C}\), which defines a vertical 2-form along the fibres of \(\ul{\mc{C}}\). By minimally coupling this to the gauge field \(a\), this can be extended to a true 2-form \(\omega_{\ul{\mc{C}}}\) on \(\ul{\mc{C}}\).  

We suppose that this data can be extended over the disk and we write \((\tilde{A}_\Sigma, \tilde{\phi})\) and \(\tilde{a}\) respectively for the corresponding extensions of \((A_\Sigma, \phi)\) and \(a\).

The action \eqref{eq:lagcsm} becomes
\begin{equation}
\label{eq:lagqm1}
S = \frac{\lambda}{2\pi} \int_{D^2} \left( (\tilde{A},\tilde{\phi})^* \omega_{\ul{\mc{C}}}  + (\nu(\tilde{A}, \tilde{\phi}) - e^2 \xi, F(\tilde{a}) )\right)- \beta \int_{S^1}  E_\text{YMH}(A, \phi) \dd t \text{,}
\end{equation}
where we recall that the moment map \(\nu\) is
\[
\nu(\tilde{A}, \tilde{\phi}) = *_\Sigma F(\tilde{A}) + e^2 \mu(\tilde{\phi}) \text{,}
\]
and have used bracket notation to denote the inner product on \(\text{Lie}(\mc{G}) = \Gamma(\text{ad}_P)\), which can be thought of as being induced by the inner product on the original Lie algebra \(\mathfrak{g}\) and integration over \(\Sigma\).

Applying Stokes' theorem to \eqref{eq:lagqm1} gives
\begin{equation}
S = \frac{\lambda}{2\pi} \int_{D^2} (\tilde{A}, \tilde{\phi})^*\omega_{\ul{\mc{C}}} + \frac{1}{2\pi} \int_{S^1}  \left( \lambda (\nu(A,\phi) - \xi, a_t)  - \beta E_\text{YMH} \right)\dd t \text{,}
\end{equation}
where \(a = a_t \dd t\).

The particular Lagrangian \eqref{eq:lagqm1} takes the standard form of gauged Hamiltonian mechanics (as in \cite{xuGHF}) in \(\mc{C}\). The gauge group is \(\mc{G}\) and the \(\mc{G}\)-invariant Hamiltonian is \(E_\text{YMH}\). The gauge field \(a\) of this mechanics enters only as a Lagrange multiplier, and integrating it out imposes the `Gauss law'
\[
* F(A) = e^2 ( \xi - \mu(\phi))\text{,}
\]
which we identify as the first vortex equation. With this done, we may quotient by the gauge action, after which the action \eqref{eq:lagqm1} becomes (in the temporal gauge, \(a = 0\))
\begin{equation}
\label{eq:lagqm2}
 \frac{\lambda}{2\pi} \int_{D^2} (\tilde{A},\tilde{\phi})^* \omega_\mc{C//\mc{G}} - \frac{\beta}{2\pi} \int_{S^1}  \left( \int_\Sigma \left(|\db_A \phi|^2 \omega_\Sigma \right)+ E_\text{top} \right) \dd t
\end{equation}
which is a Hamiltonian mechanics on the symplectic quotient \(\mc{C} // \mc{G}\), with Hamiltonian (up to the topological term, \(E_\text{top}\), which is just a ground state energy)
\[
\int_\Sigma |\db_A \phi |^2 \omega_\Sigma \text{,}
\]
which we view as a function of gauge equivalence classes of pairs \((A, \phi)\).
In the limit of low energy\footnote{In condensed matter applications, one could imagine dissapative effects driving the system towards low energy. In pure Hamiltonian mechanics, the energy is conserved.}, or low temperature (that is, rescaling \(t \mapsto \beta t\) and taking \(\beta \to \infty\)), we are led to consider the  configurations which minimise this energy, which is the locus of
\[
\db_A \phi = 0 \text{.}
\]
Imposing this equation puts on the vortex moduli space, with restricted action
\begin{equation}
\label{eq:lagvortex}
S_\text{vortex} = \frac{\lambda}{2\pi} \int_{D^2} \tilde{z}^* \omega_\mc{M} 
\end{equation}
where \(z\) is a simply-connected loop \(S^1 \to \mc{M}\), \(\tilde{z}\) is an extension of \(z\) to the disk, the 2-form \(\omega_\mc{M}\) is the K\"ahler form on the vortex moduli space, and we have subtracted off the dynamically unimportant topological term. The reason that this all works out is the fact that the vortex moduli space is a symplectic quotient.

The action \eqref{eq:lagvortex} describes Hamiltonian mechanics on the vortex moduli space with zero Hamiltonian. Notice that the vortex moduli space is the phase space of the theory, not the configuration space (as it would be in a theory with second-order dynamics). The theory is classically `trivial' (modulo the difficult problem of finding a vortex solution), with no dynamics. The vortices sit still, reflecting their BPS nature. This theory is interesting quantum mechanically, however: it is a (generally nontrivial) topological quantum mechanics on the vortex moduli space.

Note that we could have taken the low-temperature limit before integrating out the one-dimensional gauge field. In this case, we get the low-temperature gauged quantum mechanics
\[
S_{\beta \to \infty} = \frac{\lambda }{2\pi} \int_{D^2}  \left( (\tilde{A}, \tilde{\phi})^* \omega_{\ul{\mc{C}_0}} + (\nu(\tilde{A}, \tilde{\phi} ) - e^2 \xi, F(\tilde{a})) \right) \text{.}
\]
Integrating out \(a\) and taking the quotient by the group \(\mc{G}\) leaves one with the theory \eqref{eq:lagvortex}.

\subsection{Quantum vortex mechanics}

The most natural way to quantise the classical theory of \eqref{eq:lagvortex} is by \emph{geometric quantisation} (see \cite{woodhouseGQ}). This proceeds by finding a holomorphic line bundle
\[
\mc{L}^\lambda \to \mc{M}
\]
over the moduli space with first Chern class
\[
c_1(\mc{L}^\lambda) = \frac{\lambda }{2\pi} [\omega_\mc{M}] \in H^2(\mc{M}, \mathbb{Z}) \text{.}
\]
This imposes a Bohr--Sommerfeld quantisation condition on the K\"ahler form \( \omega_\mc{M}\).

The Hilbert space of states is then identified with the space of holomorphic sections of \(\mc{L}^\lambda\):
\[
\mc{H}_0(\lambda) \coloneqq H^0(\mc{L}^\lambda)\text{.}
\]
In general, even the dimension of this space is hard to access. It generally depends on the precise choice of line bundle \(\mc{L}\). A more accessible quantity is the dimension of the graded space
\[
\mc{H}(\lambda) \coloneqq \sum_i (-1)^i H^i(\mc{L}^\lambda) \text{,}
\]
The graded dimension of this is the \emph{Euler characteristic} \(\chi(\mc{L}^\lambda)\) of \(\mc{L}^\lambda\). If there is a well-defined bundle \(\mc{L}\) at \(\lambda = 1\), we may also call \(\chi(\mc{L}^\lambda)\) the \emph{Hilbert polynomial} of \(\mc{L}\) when viewed as a function of \(\lambda\).

The Euler characteristic is rendered accessible by virtue of the Hirzebruch--Riemann--Roch theorem
\begin{equation}
\label{eq:HRR}
\chi(\mc{L}^\lambda) \stackrel{!}{=} \int_\mc{M} \text{ch}(\mc{L}^\lambda) \text{td}(\mc{M})
\end{equation}
where \(\text{ch}(\mc{L}^\lambda) = \exp(\frac{\lambda}{4\pi} [\omega_\mc{M}])\) is the Chern character of \(\mc{L}^\lambda\) and \(\text{td}(\mc{M})\) is the Todd class of the tangent bundle of \(\mc{M}\). This is an entirely topological decription of the \textit{a priori} analytic quantity \(\chi(\mc{L}^\lambda)\).
The main aim of this note is to compute \(\chi(\mc{L}^\lambda)\) in the case of \(G = U(N_c)\) and \(Y = \mathbb{C}^{N_cN_f}\).

There are two subtleties which should be addressed, as follows.
\begin{itemize}
	\item \emph{Polarisation dependence}. Generally, geometric quantisation requires a choice of polarisation. This is, roughly speaking, a choice between `position' and `momentum' representations for the quantum states. In our context, the fact that the moduli space \(\mc{M}\) is K\"ahler allows us to gloss over this issue: the polarisation here is the choice of holomorphic structure on the bundle \(\mc{L}\). Polarisations of this type are called \emph{K\"ahler polarisations}.
	
	It is an interesting question to ask how the Hilbert space changes as one varies the choice of polarisation. For pure Chern--Simons theory quantised on the product of a smooth, closed two-dimensional surface \(\Sigma\) and the real line, the space of K\"ahler polarisations is the moduli space of complex structures on \(\Sigma\) and the bundle of Hilbert spaces over this space carries a projectively flat connection, the \emph{Hitchin connection} \cite{hitchinFC} (see also \cite{wittenQFT,axelrodGQ,andersenHC}). This demonstrates that, in this instance, the choice of K\"ahler polarisation is unimportant. 
	
	On the other hand, it was argued in \cite{erikssonKQ} that for quantisations of moduli spaces of local Abelian vortices it is generally not possible to construct a similar projectively flat connection on the moduli space of K\"ahler polarisations. 
	
	For us, the details of different polarisations are relatively unimportant, as we simply compute a topological invariant.  
	
	\item \emph{The metaplectic correction}. The space of quantum states should have a natural inner product. One way to ensure that this is so is to `upgrade' the quantum states to \(\mc{L}\)-valued half-densities on \(\mc{M}\), so that they pair to give a density that can be integrated over \(\mc{M}\). Mathematically, this means taking the quantum states to be sections of \(\mc{L} \otimes K_\mc{M}^\frac{1}{2}\) rather than sections of \(\mc{L}\), where \(K_\mc{M}^\frac{1}{2}\) is a square root of the canonical bundle of \(n\)-forms on \(\mc{M}\), if such a square root exists.
	
	The metaplectic correction is often necessary in the case of non-K\"ahler polarisations, as otherwise the Hilbert space that one obtains is often empty. In general though, the metaplectic correction may not even exist, as is generally the case for Chern--Simons theory \cite{axelrodGQ} and for quantisations of Abelian vortex moduli spaces \cite{erikssonKQ}. The fact that metaplectic corrections do not generally exist for the theories which we are interested in suggests that we should not consider them and so we will not attempt to include a metaplectic correction in our quantisation. 
\end{itemize}

\subsection{The path integral}

There is another approach to quantising the theory: the path integral approach. The partition function for the low-temperature theory on \(\Sigma\times S^1\) is written schematically as
\[
Z = \int_{\mc{L}_0\mc{M}} \mc{D}z \exp\left(\ii \frac{\lambda}{2\pi}\int_{S^1} z^* \theta_\mc{M} \right)  
\]
where \(\theta_\mc{M}\) is a local symplectic potential on \(\mc{M}\) and \(\mc{L}_0 \mc{M}\) is the space of contractible loops on \(\mc{M}\).

On general grounds, the partition function is the dimension of the Hilbert space of the theory. One can compute it directly using localisation techniques (see, for example, \cite{szaboEL}). Localisation allows one to reduce the path integral to an integral over the space of constant loops, which is simply \(\mc{M}\). Up to some potential topological subtleties, this integral over \(\mc{M}\) is exactly that on the right-hand-side of \eqref{eq:HRR}. 

\section{Quantisation of Abelian vortex moduli}
\label{sec:avm}

\subsection{Characteristic classes}

In the Abelian case of \(N_c = 1\), we will compute the Euler characteristic \(\chi(\mc{L}^\lambda)\) rigorously by a direct integration over the moduli space using the Hirzebruch--Riemann--Roch theorem \eqref{eq:HRR}. To do this, we need to get to grips with the integrand. Here we recall some relevant information about characteristic classes.

Let \(X\) be a complex manifold and let \(E \to X\) be a complex vector bundle or rank \(n\). If 
\[
c_t(E) = 1 + c_1(E)t + c_2(E)t^2 + \cdots + c_n(E) t^n
\]
is the Chern polynomial of \(E\), the \emph{Chern roots} \(\{\alpha_i(E)\}_{i=1, \cdots, n}\) of \(E\) are defined by
\[
c_t(E) = \prod_{i=1}^n (1+ \alpha_i(E) t)\text{.}
\]
Each of the \(\alpha_i\) has cohomological degree 2. The fact that the Chern polynomial can be written in this way is a result of the splitting principle. 

The \emph{Todd class} of \(E\) is then defined by
\[
\text{td}(E) \coloneqq \prod_{i=1}^n Q(\alpha_i(E))
\]
where
\[
Q(x) = \frac{x}{1-e^{-x}}\text{,}
\]
which we interpret in terms of a formal power series. It is immediate from this definition that the Todd class is multiplicative over exact sequences of vector bundles. 
We write \(\text{td}(X) \coloneqq \text{td}(T_X)\).

\begin{ex}[Projective space]
The tangent bundle of \(\mathbb{P}^n\) fits into the Euler exact sequence
\[
0 \to \mc{O} \to \mc{O}^{\oplus(n+1)} \to T_{\mathbb{P}^n} \to 0 \text{.}
\]	
Thus
\[
\text{td}(\mathbb{P}^n) = \left( \frac{\xi}{1-e^{-\xi}} \right)^{n+1}
\]
where \(\xi = c_1 (\mc{O}(1))\) is the generator of \(H^2(\mathbb{P}^n)\).
\end{ex}

\begin{ex}[Projective bundles]
\label{ex:projbundle}
Let \(V \to Y\) be a vector bundle of rank \(r\), and let \(X \coloneqq \mathbb{P}(V) \xrightarrow{p} Y\) be the corresponding projective bundle, given by projectivising the fibres of \(V\). The space \(X\) carries the natural tautological line bundle \(\mc{O}_X(1)\). 

There are two useful exact sequences, the relative Euler sequence
\[
0 \to \mc{O}_X \to p^*V \otimes \mc{O}_X(1) \to T_{X/Y} \to 0
\]
and the (holomorphic) Atiyah sequence
\[
0 \to T_{X/Y} \to T_X \to p^* T_Y \to 0 \text{.}
\]
These reveal that
\[
\text{td}(X) =   \text{td}(p^*V \otimes \mc{O}_X(1)) p^*\text{td}(Y) \text{.}
\]
\end{ex}

\subsection{Local Abelian vortices}

Let us begin in the local Abelian case of \(N_c = N_f = 1\). In this case, our discussion retreads some of the ideas of \cite{macdonaldSP} on the cohomology rings of symmetric products of Riemann surfaces. The results of \cite{macdonaldSP} were used in \cite{deyGQ} to compute the expected dimension of the Hilbert space of the vortex quantum mechanics in the special case \(\lambda = 1\) and in \cite{mantonVVM} to compute the volume of the moduli space of local Abelian vortices. 

The moduli space of local Abelian vortices on a Riemann surface \(\Sigma\) of genus \(g\) is well-known to be the space of effective divisors on \(\Sigma\) \cite{taubesCI,noguchiYMH,bradlowVLB,garciapradaDEP}. This can be thought of as the space of holomorphic line bundles with normalised holomorphic section, which facilitates the following construction of the vortex moduli space as a complex manifold\footnote{Another, direct, way to think of vortices in these terms follows from looking at the dissolving vortex limit \cite{wehrheimVI,baptistaDV}, a kind of (shifted) weak-coupling limit, where the vortex equations explicitly become the equations for a line bundle with a flat connection with normalised holomorphic section.}.

\begin{figure}[ht]
\centering
\begin{tikzpicture}[scale = 0.8]
  \useasboundingbox (-3,-1.5) rectangle (3,4);
  \draw (-2,3) circle (1cm) node[anchor=south]{\(\mathbb{C}P^{k-g}\)};
  \draw (-3,3) arc (180:360:1 and 0.3);
  \draw[dashed] (-2,2) -- (-2,0.5);
  \draw[fill=black] (-2,0.5) circle (0.025);
  \draw (0,0) ellipse (3 and 1.5) node[anchor=north, yshift=-0.5cm]{\(P^k(\Sigma)\)};
    \clip (0,-1.8) ellipse (3 and 2.5);
    \draw (0,2.2) ellipse (3 and 2.5);
    \clip (0,2.2) ellipse (3 and 2.5);
    \draw (0,-2.2) ellipse (3 and 2.5);
\end{tikzpicture}	
\caption{A sketch of the moduli space of charge \(k\) local Abelian vortices for \(k > 2g-2\). The moduli space fibres over the torus \(P^k(\Sigma)\), which is the moduli space of charge \(k\) holomorphic line bundles, with fibre \(\mathbb{C}P^{k-g}\), which is the space of normalised sections.  When \(k \leq 2g-2\), there are degenerate fibres (see \cite{mantonVJV}).} 
\label{fig:abmoduli}
\end{figure}

Write \(P^k(\Sigma)\) for the moduli space of holomorphic line bundles on \(\Sigma\) of degree \(k\). This is a torus of complex dimension \(g\). Let 
\[
\mc{U} \to P^k (\Sigma) \times \Sigma
\]
be the universal degree \(k\) line bundle. This is well-defined only after choosing a point \(x \in \Sigma\) and asking that \(\mc{U}|_{P^k(\Sigma) \times \{x\}} \) be trivial. 

Writing \(q : P^k(\Sigma) \times \Sigma \to P^k(\Sigma)\) for the projection, we consider the direct image sheaf \(q_!\mc{U}\) over \(P^k(\Sigma)\). If \(k > 2g - 2\), this is locally free of constant rank and defines a vector bundle \(V\) over \(P^k(\Sigma)\) with fibre over a holomorphic line bundle \(\mc{L}\) the space of holomorphic sections of \(\mc{L}\). The moduli space of vortices is then
\[
\mc{M} \cong \mathbb{P}(V) \text{,}
\]
which admits a projection map \( p : \mc{M} \to P^k(\Sigma)\) and supports the tautological line bundle \(\mc{O}_\mc{M}(1)\).

Combining this with \cref{ex:projbundle} tells us that the Todd class of \(\mc{M}\) is
\[
\text{td}(\mc{M}) = \text{td}(p^*V \otimes \mc{O}_\mc{M}(1)) p^*\text{td}(P^k(\Sigma)) \text{.}
\]
The space \(P^k(\Sigma)\) is a torus and so has trivial tangent bundle, so the Todd class becomes
\[
\text{td}(\mc{M}) = \text{td}(p^*V \otimes \mc{O}_\mc{M}(1))\text{.}
\]

To compute the Chern roots of \(V = q_! \mc{U}\), we use the Grothendieck--Riemann--Roch theorem. We have
\begin{align}
\text{ch}(V) &= \text{ch}(q_! \mc{U}) \nonumber \\
	&\stackrel{!}{=} q_* \left( \text{ch}(\mc{U}) \text{td}(\Sigma) \right) \text{,} \label{eq:GRR}
\end{align}
where the first equality follows from the definition of \(V\) and the second follows from the Grothendieck--Riemann--Roch theorem.
Write \(\hat{\omega}\) for the generator of \(H^2(\Sigma,\mathbb{Z})\). Then
\[
\text{td}(\Sigma) = 1 + (1-g)\hat{\omega}
\]
and
\[
c_1(\mc{U}) = k \hat{\omega} + t
\]
where \(t \in H^1(\Sigma) \otimes H^1(P^k(\Sigma))\) is the class dual to the natural evaluation map.

Using the fact that \(\text{ch}(\mc{U}) = \exp(c_1(\mc{U}))\) and the Grothendieck--Riemann--Roch formula \eqref{eq:GRR}, we see that
\begin{align*}
\text{ch}(V) =  k + 1 -g + \frac{1}{2} q_* t^2 \text{.}
\end{align*}
The class \(\sigma = - \frac{1}{2} q^* t^2 \in H^2(P^k (\Sigma))\) is that of the theta divisor.

This reveals that \(\text{rank}(V) = k + 1 -g\), that \(c_1(V) = - \sigma\), and that all of the higher Chern classes conspire to cancel the higher contributions to the Chern character. This means that
\[
c_i (V) = \frac{1}{i!} (-\sigma)^i \text{,}
\]
and the Chern polynomial is \(c_t(V) = e^{-\sigma t}\) (note that this truncates because \(\sigma^{g+1} = 0\) for dimensional reasons). 

The projective bundle formula reveals that the cohomology ring of \(\mc{M}\) takes the form
\[
H^\bullet(\mc{M}) = H^\bullet(P^k(\Sigma))[\xi]/\mc{R}
\]
where the relation \(\mc{R}\) is
\[
\sum_{i=0}^g  \frac{1}{i!} (-\sigma)^i \xi^{k-g+1-i} = 0 \text{.}
\]

Write \(\sigma_i\), \(i =1, \cdots, g\) for the Chern roots of \(V\). These obey \(\sigma_i^2 = 0\) for all \(i\).
Then the Todd class of \(\mc{M}\) is
\begin{align*}
\text{td}(\mc{M}) &= \text{td}(p^* V \otimes \mc{O}_\mc{M}(1) ) \\
	&= \left( \frac{\xi}{1-e^{-\xi}} \right)^{k - 2g + 1} \prod_{i=1}^g \frac{\xi - \sigma_i}{1 - e^{-(\xi - \sigma_i)}} \\
	&= \left( \frac{\xi}{1-e^{-\xi}} \right)^{k - g + 1} e^{-\sigma X}
\end{align*}
where, formally, \(X = \frac{1}{\xi} - \frac{e^{-\xi}}{1-e^{-\xi}}\) and we interpret \(e^{-\sigma X}\) as a formal power series which truncates.

The K\"ahler class of the moduli space is
\[
\frac{1}{2\pi} [\omega_\mc{M}] = d \xi + \sigma 
\]
where \(d = \mc{A} - k\). To see this, one can write out representatives for \(\xi\) and \(\sigma\) directly (see \cite{perutzSF,baptistaL2}, for example).

Let \(\mc{L}\) be a holomorphic line bundle with \(c_1(\mc{L}) = d \xi + \sigma\). 
There is a natural choice for this quantum line bundle \(\mc{L}\), namely
\[
\mc{L} = \mc{O}_\mc{M}(d) \otimes p^* \text{det}(V)\text{,}
\]
although we do not insist on this choice (indeed, \(V\) was defined with respect to an arbitrary choice of normalisation for the Poincar\'e line bundle, so the bundle above is not canonically determined).

Ideally, we would like to compute 
\[
h^0(\mc{L}^\lambda) = \text{dim} H^0(\mc{L}^\lambda)
\]
where \(\lambda\) is a parameter (physically, it is the Chern--Simons level). This is generally not accessible, and may depend on the particular choice of \(\mc{L}\). Instead, we would like to compute the Euler characteristic
\[
\chi (\mc{L}^\lambda) = \sum_{i} (-1)^i h^i(\mc{L}^\lambda) \text{.}
\]
The Riemann--Roch theorem tells us that
\[
\chi ( \mc{L}^\lambda ) = \int_\mc{M}\text{ch}(\mc{L}^\lambda) \text{td}(\mc{M})
\]
which is
\[
\int_\mc{M} e^{\lambda ( d \xi + \sigma )} \left( \frac{\xi}{1-e^{-\xi}} \right)^{k - g + 1} e^{-\sigma X} \text{.}
\]
This becomes
\begin{align*}
\chi(\mc{L}^{\lambda} ) &= \sum_{j=0} \frac{\lambda^j}{j!} \int_\mc{M} \sigma^j e^{\lambda d \xi} 	\left( \frac{\xi}{1-e^{-\xi}} \right)^{k - g + 1} e^{-\sigma \left( \frac{1}{\xi} - \frac{e^{-\xi}}{1-e^{-\xi}} \right)} \\
&= \sum_{j=0} \frac{\lambda^j}{j!} \int_\mc{M} \sigma^j e^{\lambda d \xi} 	\left( \frac{\xi}{1-e^{-\xi}} \right)^{k - g + 1} e^{\sigma \left(\frac{e^{-\xi}}{1-e^{-\xi}} \right)}
\end{align*}
using the relation \(\mc{R}\). Expanding the exponential, collecting the powers of \(\sigma\), and removing terms that do not contribute to the integral, we find
\begin{align*}
\chi(\mc{L}^{ \lambda} ) &= \sum_{j=0}^g \frac{\lambda^j}{j! (g-j)!} \int_\mc{M} \sigma^g \xi^{k-g+1} \frac{e^{(\lambda d - g + j)\xi} }{(1-e^{-\xi})^{k-j+1}} \\
&= \sum_{j=0}^g \lambda^j {g \choose j} \text{Res}_{x=0} \frac{e^{(\lambda d - g + j )x}}{(1-e^{-x})^{k-j+1}} \dd x \text{,}
\end{align*}
where we have used that the integral of \(\sigma^g\) over the Jacobian is \(g!\), which is a classical result. 

To compute a residue of the form
\[
\text{Res}_{x=0} \frac{e^{px}}{(1-e^{-x})^{q+1}} \dd x
\]
one may introduce a variable \(y\) by \(1 - e^{-x} = y\). Then \(e^x = (1-y)^{-1}\) and \(\dd x = \frac{1}{1-y} \dd y\). Then the residue is
\[
\text{Res}_{y=0} \frac{(1-y)^{-p-1}}{y^{q+1}} \dd y = {p + q \choose q}\text{.}
\]

Putting the pieces together gives us the following result.

\begin{theorem}
\label{theorem:localab}
Let \(\mc{L} \to \mc{M}\) be the quantum line bundle over the moduli space of local Abelian vortices of charge \(k\) on a closed Riemannian surface \(\Sigma\) of genus \(g\) and dimensionless `area' \(\mc{A} = \frac{e^2 \tau \text{vol}(\Sigma)}{4\pi} \in \mathbb{Z}\). Suppose that \(k > 2g-2\). Then the Hilbert polynomial of \(\mc{L}\) is
\[
\chi(\mc{L}^{\lambda} ) =  \sum_{j = 0}^{g} \lambda^j {g \choose j}{\lambda( \mc{A} - k) + k - g \choose k-j} \text{.}
\]	
\end{theorem}

\begin{remark}
Our proof of this result only holds literally in the case that \(k \geq 2g-1\). However, the method still goes through in a virtual sense for all \(k\), and the result holds in general. Note that if \(k < g\), the terms with \(j = k +1, \cdots, g \) all vanish.
\end{remark}

\begin{remark}
One can express the Euler characteristic as a single residue as
\[
\chi(\mc{L}^{\otimes \lambda} ) = \text{Res}_{x=0} \left( e^{-x} + \lambda (1-e^{-x})\right)^g \frac{e^{\lambda d x}}{(1-e^{-x})^{k+1}} \dd x \text{.}
\] 
This is the form that one might find the answer in if one exploited Coloumb branch localisation. This makes particularly clear the simplification that occurs if \(\lambda = 1\).
\end{remark}

This theorem mediates between two known results. Setting \(\lambda = 1\) computes the Euler characteristic of the quantum line bundle \(\mc{L}\), which is the expected dimension of the Hilbert space of the vortex quantum mechanics at Chern--Simons level 1. This is
\begin{align*}
\chi(\mc{L}) &= \sum_{j=0}^{g} {g \choose j} {\mc{A} \choose k - j } \\
	&= {\mc{A} \choose k} \text{,}
\end{align*}
as previously found in \cite{erikssonKQ} and in \cite{deyGQ}. In the full quantum theory, one should sum over topological sectors with a fugacity \(t\) for the topological symmetry. The resulting \emph{grand canonical partition function} is
\begin{align*}
Z(t) &= \sum_{k=0}^{\infty} {\mc{A} \choose k} t^k \\
	&= (1+t)^\mc{A} \text{.}
\end{align*}

This result exhibits a kind of `vortex-hole duality'. By virtue of the Bradlow bound, a configuration of \(k\) vortices in a normalised area \(\mc{A}\) might also be viewed as a configuration of \(\mc{A} - k\) vortex holes in a sea of \(\mc{A}\) vortices. This is reflected in the vortex quantum mechanics by the equality
\[
{\mc{A} \choose k} = {\mc{A} \choose \mc{A} - k} \text{.}
\]

On the other hand, the volume of the vortex moduli space is (up to constant factors)
\[
\text{vol}(\mc{M}) = \int_\mc{M} e^{[\omega_\mc{M}]} \text{.}
\]
Because the leading term in the Todd class is unity, this is the coefficient of \(\lambda^{\text{dim}(\mc{M})}\) in \(\chi(\mc{L}^{ \lambda})\) (this is the highest power of \(\lambda\) that appears). We can read this off to be
\[
\sum_{j=0}^{\text{min}(k,g)} {g \choose j} \frac{(\mc{A} - k)^{k-j}}{(k-j)!}
\]
in agreement with the computations of \cite{mantonVVM}, \cite{miyakeVMS} and \cite{ohtaHCB}. 

\subsection{Semilocal Abelian vortices}

Introducing more flavours does not significantly alter the nature of the above calculation. As revealed by the dissolving vortex limit, the moduli space is again a projective bundle. Indeed, if the number of flavours is \(N_f\), the moduli space of charge \(k\) vortices is, for \(k\) sufficiently large,
\[
\mc{M}_{N_f} = \mathbb{P}( V^{\oplus N_f})
\]
where \(V \to P^k(\Sigma)\) is the vector bundle introduced previously. The dimension of the moduli space is \(N_f(k-g+1) + g - 1\), in agreement with an index calculation.

Multiplicativity of the Chern class means that
\begin{align*}
c(V^{\oplus N_f}) &= c(V)^{N_f} \\
	&= e^{- N_f \sigma} \text{.}
\end{align*}
Again, the formal power series implied by \(e^{-N_f \sigma}\) truncates because \(\sigma^{g+1} = 0\).

The cohomology ring of the moduli space can be described using the projective bundle formula to be
\[
H^\bullet (\mc{M}_{N_f}) = H^\bullet(P^k(\Sigma)) [\xi]/ \mc{R}
\]
where the relation \(\mc{R}\) is
\[
\sum_{i=0}^g \frac{1}{i!} (-N_f \sigma)^i \xi^{N_f(k-g+1) - i} = 0\text{.} 
\]

The Todd class of \(\mc{M}_{N_f}\) is 
\[
\text{td}(\mc{M}_{N_f}) = \left( \frac{\xi}{1-e^{-\xi}} \right)^{N_f(k-g+1)} e^{-N_f \sigma X}\text{,}
\]
where \(X = \frac{1}{\xi} - \frac{e^{-\xi}}{1-e^{-\xi}}\) as before. 

The cohomology class of the \(L^2\) K\"ahler form on the vortex moduli space is (just as before)
\[
\frac{1}{2\pi} [\omega_{\mc{M}_{N_f}}] = d \xi + \sigma \text{,}
\]
where \(d = \mc{A} - k\).
(Note that the second cohomology group of the moduli space has two generators, so all that one needs to do to verify this is to check the coefficients.)

We again compute the Euler characteristic of a quantum line bundle \(\mc{L} \to \mc{M}_{N_f}\) with \(2\pi c_1(\mc{L}) = [\omega_{\mc{M}_{N_f}}]\). There is again a natural choice, given by
\[
\mc{L} = \mc{O}_{\mc{M}_{N_f}}(d) \otimes p^*\text{det}(V)\text{,}
\]
although, again, we do not insist on this, noting that it is not canonically determined.

Going through the motions as we did for the case of \(N_f = 1\) leads to the result
\begin{align*}
\chi(\mc{L}^{\lambda} ) &= \sum_{j=0}^g \lambda^jN_f^{g-j} {g \choose j} \text{Res}_{x=0} \frac{e^{(\lambda d - (g-j) )x}}{(1-e^{-x})^{N_f(k-g+1) + g - j}}  \dd x\\
	&= \sum_{j=0}^g \lambda^j N_f^{g-j} {g \choose j}{\lambda d + N_f(k-g +1) - 1 \choose N_f(k - g + 1) +g-j-1}\text{.}
\end{align*}
Substituting in \(d = \mc{A} - k\) and introducing the notation \(\delta = \text{dim}_\mathbb{C}(\mc{M}_{N_f})\) leads to the following theorem.

\begin{theorem}
\label{theorem:semilocalab}
Let \(\mc{M}_{N_f}\) be the moduli space of charge \(k\) vortex configurations in a \(U(1)\) gauge theory with \(N_f\) fundamental flavours on a closed Riemannian surface \(\Sigma\) of genus \(g\) and dimensionless `area' \(\mc{A} = \frac{e^2 \tau \text{vol}(\Sigma)}{4\pi} \in \mathbb{Z}\), supposing that \(k > 2g-2\). Let \(\mc{L} \to \mc{M}_{N_f}\) be a quantum line bundle for the \(L^2\) vortex K\"ahler form on \(\mc{M}_{N_f}\). Then
\[
	\chi(\mc{L}^\lambda) = \sum_{j=0}^g \lambda^j N_f^{g-j} {g \choose j} {\lambda (\mc{A} - k) + \delta -g \choose \delta-j}\text{,}
\]
where \(\delta = N_f(k-g+1) + g - 1\) is the dimension of the moduli space.
\end{theorem}

\begin{remark}
While our proof only literally holds for \(k\) sufficiently large, the argument should go through in a virtual sense in the general case.\end{remark}

\begin{remark}
Again, one can express the Euler characteristic as a single residue. Indeed, one has
\[
\chi(\mc{L}^{\lambda}) = \text{Res}_{x=0} (N_f e^{-x} + \lambda (1- e^{-x}))^g \frac{e^{\lambda d x}}{(1-e^{-x})^{\delta +1}}  \dd x\text{,}
\]
which simplifies in the case that \(\lambda = N_f\).
\end{remark}

As in the local case, there are two particularly interesting corollaries of this result. Setting \(\lambda = N_f\), one finds the visually simple result
\begin{align*}
\chi(\mc{L}^{ N_f} ) &= N_f^g \sum_{j=0}^g {g \choose j}{ N_f(\mc{A} -k) + \delta - g \choose \delta - j} \\
	&= N_f^g { N_f(\mc{A} - k) + \delta \choose \delta} \text{.}
\end{align*}
We interpret this result in \autoref{sec:slduality}.

On the other hand, identifying the coefficient of the highest power of \(\lambda\) in \(\chi(\mc{L}^{ \lambda})\) gives us the volume of the vortex moduli space, up to factors of \(2\pi\). We read this off to be
\[
\sum_{j=0}^{\text{min}(\delta,g)} N_f^{g-j} {g \choose j} \frac{(\mc{A} - k)^{\delta-j}}{(\delta-j)!} \text{,}
\]
in agreement with the result of \cite{miyakeVMS,ohtaHCB}.

\subsection{Abelian vortices in a harmonic trap}
\label{subsec:harmonictrap}

There is a special class of deformations of the theory that `preserve critical coupling'. These are deformations induced from Hamiltonian actions of groups on the configuration space of the theory preserving the vortex equations.  Such deformations induce a potential on the moduli space for the Hamiltonian mechanics controlling the low-temperature behaviour of the theory. The potential is such that the Hamiltonian flow generates the symmetry.

Here we consider an example in the local Abelian case, that of a harmonic trap. The symmetry generating this is a rotational symmetry of space, so it requires that \(\Sigma\) has a rotational symmetry. We will take \(\Sigma\) to be the round sphere. This has previously been considered in \cite{romaoGVB}. In the case that \(\Sigma\) is the plane, the analogous story was considered in \cite{tongVMHT} and plays an important role in the dual approach to quantum Hall physics on the plane \cite{tongQHF} (as, in that case, there is no Bradlow bound and a different mechanism is necessary to induce the creation of a region of Coulomb phase). 

Suppose that \(\Sigma = S^2\) with a round metric and consider \(U(1)\) gauge theory with one fundamental flavour. The moduli space of \(k\)-vortex solutions is then the complex projective space \(\mc{M} = \mathbb{C}P^k\). A useful coordinate system for the moduli spaces can be given as follows. Let \(z\) be a complex coordinate on \(\Sigma \cong \mathbb{C}P^1\). A \(k\)-vortex configuration is uniquely specified by \(k\) points \((z_1, \cdots, z_k)\) on \(\Sigma\) and so we may parameterise the moduli space (away from configurations with vortices at \(z = \infty\)) with the polynomials
\[
\prod_{i=1}^k (z - z_i) \text{.}
\]
Expanding this polynomial gives an alternative parameterisation in terms of (an affine patch of) homogeneous coordinates \([w_0 : w_1 : \cdots : w_k]\) as
\[
\prod_{i=1}^k (z - z_i) = w_0 z^k + w_1 z^{k-1} + \cdots + w_k \text{.} 
\]
In this patch, \(w_0 = 1\), \(w_1 = \sum_i z_i\), and so on. 

The theory on the round sphere \(S^2\) has a global \(U(1)\) symmetry given by rotating the sphere around a fixed axis. This induces an action of \(U(1)\) on the space \(\mc{C}\) of vortex configurations in the obvious way, rotating the configuration. This action preserves the vortex equations and so descends to the moduli space of vortices.

We may align our complex coordinate \(z\) on \(\Sigma\) in such a way that \(e^{\ii\theta} \in U(1)\) acts via
\[
z \mapsto e^{\ii \theta} z \text{,}
\]
so that \(z = 0\) and \(z= \infty\) are the fixed points of the action. This action induces the action
\[
z_j \mapsto e^{\ii\theta}z_j
\]
for \(j =1,\cdots,k\) on the moduli space coordinates, which in turn implies that
\begin{equation}
\label{eq:groupaction1}
w_j \mapsto e^{\ii \theta j}w_j 
\end{equation}
for \(j = 1, \cdots, k\).

We can now carry out an equivariant geometric quantisation of the moduli space. Again, we will consider the simpler question of computing the equivariant index. See \cite[Section 2]{pestunLT} for an introduction to these ideas, which we summarise here. 

The quantum line bundle is \(\mc{O}(d)^\lambda \to \mathbb{C}P^k\), which lifts to an equivariant line bundle. We are interested in computing the equivariant Euler characteristic
\[
\chi_{U(1)}(\mc{O}(d)^\lambda ; q) = \sum_i (-1)^i\text{tr}_{ H^i(\mathbb{C}P^k, \mc{O}(d)^\lambda )} (q)
\]
of this equivariant line bundle. Here, for \(V\) a representation of \(U(1)\), \(\text{tr}_{V}(q) \) is the trace of \(q \in U(1)\) in the representation \(V\). If \(U(1)\) acts trivially, then any element of \(U(1)\) is represented as the identity and the trace recovers the dimension of the vector space, so that the equivariant Euler characteristic is the usual Euler characteristic in that case.

One way to compute this is using the equivariant Hirzebruch--Riemann--Roch theorem, which realises the equivariant Euler characteristic as an integral of an element of the equivariant cohomology of the moduli space. One has
\[
\chi_{U(1)}(\mc{O}(\lambda d) ; q) = \int_{\mathbb{C}P^k} \left( \text{ch}_{U(1)} (\mc{O}(\lambda d)) \text{td}_\text{U(1)} (\mathbb{C}P^k) \right)(q) \text{,}
\]
where \(\text{ch}_{U(1)} (\mc{O}(\lambda d) )\) is the equivariant Chern character of the bundle \(\mc{O}(\lambda d)\) and \(\text{td}_{U(1)} (\mathbb{C}P^k) \) is the equivariant Todd class of \(T_{\mathbb{C}P^k}\). 

As recalled in \cite{pestunLT,szaboEL}, this integral can be reduced to a sum of contributions from the fixed points of the group action. The outcome of this (when the fixed point set is discrete, which it is for us) is the Lefshetz formula
\[
\chi_{U(1)} (\mc{O}(\lambda d); q) = \sum_{p \in F}  \frac{\text{tr}_{\mc{O}(\lambda d)_p} (q)}{\text{det}_{T^{1,0}_{p}{\mathbb{C}P^k}}(1-q^{-1})} \text{.}
\]

We can apply this to the group action \eqref{eq:groupaction1}. This action has \(k+1\) fixed points \((p_0, \cdots, p_k)\). Physically, the fixed point \(p_i\) corresponds to \(i\) vortices sitting at \(z = 0\) and \(k-i\) vortices sitting at \(z=\infty\). The character of the \(U(1)\) representation \(\mc{O}(\lambda d)_{p_i}\) is given by
\[
\text{tr}_{\mc{O}(\lambda d)_{p_i}} (q) = q^{\lambda d i}
\]
and the determinant is
\[
\text{det}_{(T_{\mathbb{C}P^k})_{p_i}}(1-q^{-1}) = \prod_{j\neq i} (1-q^{j-i}) \text{.}
\]

The index is then
\[
\chi_{U(1)} (\mc{O}(\lambda d);q) =  \sum_{i =0}^k \frac{q^{i \lambda d}}{\prod_{ j \neq i} (1- q^{j-i})} 
\]
(see also \cite[section 2.9]{pestunLT} for a very similar computation). Rearranging this sum gives
\begin{align*}
\chi_{U(1)} (\mc{O}(\lambda d);q) &= \prod_{i=0}^k \frac{ (1-q^{i+\lambda d}) }{(1-q^i)} \\
	&\eqqcolon {\lambda d + k \choose k}_q \text{.}
\end{align*}
This is the \emph{\(q\)-binomial coefficient}, sometimes known as the \emph{Gaussian binomial coefficient}. It has the property that, for any \(n\) and \(k\),
\[
\lim_{q\to1} {n \choose k }_q = {n \choose k}
\]
so that we recover our previous result \autoref{theorem:localab} in the genus zero case.

Substituting in \(d = \mc{A} - k\), we have at \(\lambda = 1\),
\[
\chi_{U(1)} (\mc{O}( d);q) = {\mc{A} \choose k}_q \text{.}
\]

\section{Quantisation of nonAbelian vortex moduli}
\label{sec:navm}

\subsection{Coulomb branch localisation of the Euler characteristic}

After localising with respect to time translations and integrating out the time component of the gauge field, the (super) partition function of our theory is
\begin{equation}
\label{eq:eulercoulomb}
\chi(\mc{L}^\lambda) = \int_\mc{M} \text{td}(\mc{M}) \text{ch}(\mc{L}^\lambda) \text{,}
\end{equation}
where \(\mc{L} \to \mc{M}\) is a quantum line bundle on the moduli space of vortices.
When the gauge group is nonAbelian, we do not have a generators-and-relations description of the cohomology ring of the moduli space \(\mc{M}\), so it is hard to compute this integral directly. Instead, we will exploit the Jeffrey--Kirwan--Witten method of equivariant localisation (or `Coulomb branch localisation') to compute this integral. We do not give a complete account of this technique here, instead giving a fairly high-level exposition of the computation in the context relevant to us. We refer the reader to the original works \cite{jeffreyLNA,jeffreyLQC,witten2DGT} and to \cite{miyakeVMS,ohtaHCB}, which give a clear exposition of the technique in the context of computing the volumes of vortex moduli spaces. Our computation of \(\chi(\mc{L}^\lambda)\) generalises the results of these latter works.

The moduli space \(\mc{M}\) is the formal symplectic quotient of the space \(\mc{C}_0 = \{(A_\Sigma,\phi) \mid \db_A \phi = 0\} \subset \mc{C} = \{(A_\Sigma,\phi)\}\) by the group \(\mc{G}\) of gauge transformations. The lift of the tangent bundle of \(\mc{M}\) to \(\mc{C}_0\) is the equivariant virtual bundle
\[
 \text{ad}_\mc{G} \to T \mc{C}_0 \to \text{ad}_\mc{G}^\vee
\]
where \(\text{ad}_\mc{G}\) is the equivariant bundle \(\mc{C}_0 \times \text{Lie}(\mc{G}) \to \mc{C}_0\), carrying the adjoint action on the fibres. Here the first arrow is the derivative of the group action and the second arrow is the derivative of the moment map. The equivariant Todd class of this virtual bundle is
\[
\text{td}_\mc{G} (\mc{C}_0) \text{td}^{-1} _\mc{G} (\text{ad}_\mc{G} \oplus \text{ad}_\mc{G}^\vee) \text{.}
\]
In general, we must regard \(\mc{C}_0\) itself as a derived space, as the space of solutions to \(\db_A \phi = 0\) does not have constant dimension as \(A\) varies, so that \(T\mc{C}_0\) is itself a graded bundle.

The lift of the quantum bundle \(\mc{L}\) to \(\mc{C}_0\) is an equivariant line bundle \(\hat{\mc{L}}^\lambda\) with equivariant first Chern class equal to 
\[
c_1^\mc{G} (\hat{\mc{L}}^\lambda ) = \frac{\lambda}{2\pi} [\omega_{\mc{C}_0}^\mc{G} ] \in H^2_\mc{G} (\mc{C}_0) \text{,}
\]
where \(\omega_{\mc{C}_0}^\mc{G} = \omega_{\mc{C}_0} - \nu\) is the equivariant symplectic form (see \cite{atiyahMM}). 

The super partition function is then the integral
\[
\int_{ \mc{C}_0 \times \text{Lie}(\mc{G})} \left( \text{td}_\mc{G} (\mc{C}_0) \text{td}^{-1} _\mc{G} (\text{ad}_\mc{G} \oplus \text{ad}_\mc{G}^\vee) \text{ch}_\mc{G} (\hat{\mc{L}}^\lambda) \right)(\Phi) \, \mc{D} \Phi \text{.}
\]
Here \(\Phi\) is valued in the Lie algebra \(\text{Lie}(\mc{G})\) and \(\mc{D}\Phi\) is a measure on this space. Physically, this could be derived by computing the (super) partition function of the three-dimensional theory after localising with respect to time translations. The field \(\Phi\) is then the constant mode of the time-component of the gauge field.

We will use the notion of equivariant integration to localise this to an integral over the Lie algebra of the maximal torus of the group \(G\). This integral can then be expressed as a residue.

To do this, we note that the space \(\mc{C}_0 \times \text{Lie}(\mc{G})\) carries an infinitesimal \(U(1)\) action, generated by the variations 
\[
\delta((A_\Sigma, \phi), \Phi) = ((\dd_{A_\Sigma} \Phi, \Phi(\phi)), 0)\text{.}
\]
Our integrand is invariant under this action and we may localise with respect to it. The space of fixed points decomposes into a number of branches, including the Higgs branch, where \(\Phi = 0\), and the Coulomb branch, where \(\phi = 0\). The space of fixed points carries an action of the  torus \(T\cong U(1)^{N_c}\) of global Abelian gauge transformations. The fixed point locus of this \(T\) action is the Coulomb branch, where \(\phi = 0\) and \(\dd_{A_\Sigma} \Phi = 0\).

A generic point on the Coulomb branch defines a holomorphic vector bundle of the form
\[
E = \bigoplus_a L_a \to \Sigma
\]
where \(L_a\) has degree \(k_a\) and \(\sum_a k_a = k\), the total vortex number. Deformations of these \(T\)-invariant solutions within the space of configurations are given by pairs consisting of a holomorphic section \(\psi\) of \(E^{N_f}\) and a deformation of the holomorphic structure of \(E\), which is an element of 
\[
H^1(E\otimes E^\vee) = H^1\left(\mc{O}^{N_c} \oplus \bigoplus_{a \neq b} L_a \otimes L_b^{-1} \right) \text{.}
\]
Such a deformation remains \(T\)-invariant if \(\psi = 0\) and if the variation of holomorphic structure lies in \(H^1(\mc{O}^{N_c})\). 

Recall from above that the equivariant virtual bundle we need to consider is
\[
\text{ad}_\mc{G} \to T\mc{C}_0 \to \text{ad}_\mc{G}^\vee \text{.}
\]
Over a connected component of the \(T\)-invariant locus, the moment map is constant (as \(\phi = 0\) and by the quantisation of the Abelian flux). The image of the second map is therefore zero. We are left with a \(T\)-equivariant bundle whose fibre at the bundle \(E \in F\) is
\[
V_F = \left( H^0(E\otimes E^\vee) \to T_E F \oplus (\nu_F)_E \to 0 \right)
\]
where the bundle on the left is the bundle of infinitesimal automorphisms of a point \(E\) on \(F\) (that is, a residual gauge transformation) and \(\nu_F\) is the equivariant normal bundle to \(F\) in \(\mc{C}_0\), which takes the form
\[
\nu_F = \left( 0 \to H^1(\bigoplus_{a \neq b} L_a \otimes L_b^{-1}) \oplus H^0(E^{N_f}) \to H^1(E^{N_f})  \right)
\]
where we recall that we are working with the derived locus of the space of solutions to \(\db_A \phi = 0\).

The space \(F\) is torus of complex dimension \(N_c g\), with trivial tangent bundle. Putting this together, the class of \(V_F\) in equivariant K-theory is
\[
[V_F] = [T F] + \sum_a [H^\bullet (L_a^{N_f})] - \sum_{a\neq b} [H^\bullet(L_a \otimes L_b^{-1} ) ] \text{.}
\]
where by \(H^\bullet\) we mean the graded space \(H^0 - H^1\).

We now compute the Todd class of each of these summands divided by the Euler class, include the contribution of the Chern character of the quantum line bundle restricted to the component, and sum over the components of the fixed point locus. Let \(x_a, a = 1, \cdots, N_c\) be standard coordinates on the Lie algebra of \(T\).  From \autoref{sec:avm}, we know that the first and second summands integrated over \(F\) give a contribution of
\[
\prod_a \left(N_f  \frac{e^{-x_a}}{1-e^{-x_a}} + \lambda  \right)^g \frac{1}{(1-e^{-x_a} )^{N_f(k_a-g+1) } } 
\]
to the residue. The second summand gives a contribution of
\[
\prod_{a\neq b}\left(1-e^{-(x_a - x_b)}\right)^{k_a - k_b - g +1} \text{.}
\]
The differences \(x_a - x_b\) are the roots of \(\mathfrak{u}(N_c)\).

The restriction of the equivariant Chern character of the quantum bundle to the fixed point locus is
\[
\prod_a e^{\lambda d_a x_a}
\]
where \(d_a = \mc{A} - k_a\) is the action  \( - \int_\Sigma \left( F(A)_a - e^2(\tau - \mu(\phi))_a\omega_\Sigma\right)\)  when \(\phi = 0\).

Putting these pieces together and summing over the partitions of \(k\), which correspond to the fixed point loci of the \(T\)-action, we come to one of our main results:
\begin{equation}
\label{eq:residue1}
\begin{aligned}
\frac{(-1)^\sigma}{N_c!}
 \chi(\mc{L}^\lambda) = \sum_{\sum k_a = k} & \text{Res}_{\bsym{x}=0} \Bigg[ \left(\prod_{a<b} e^{x_a} e^{x_b}(e^{-x_a} - e^{-x_b})^2 \right)^{1-g} \\
 	& \times \prod_a (N_f e^{-x_a} + \lambda (1-e^{-x_a}))^g \frac{e^{\lambda d_a x_a}}{(1-e^{-x_a})^{\delta_a + 1} } \dd x_a \Bigg]
 \end{aligned}
\end{equation}
where \(d_a = \mc{A} - k_a\) and \(\delta_a = N_f ( k_a - g +1) + g-1\), the sum is over partitions of \(k\) of length \(N_c\), and
where \(\sigma = (N+1)(k+1)\). We comment on precisely how we mean to take this residue below.

Introducing \(y_a = 1-e^{-x_a}\), the residue \eqref{eq:residue1} can alternatively be written as
\begin{equation}
\label{eq:residue2}
\begin{aligned}
\frac{(-1)^\sigma}{N_c!} \sum_{\sum k_a = k} \text{Res}_{\bsym{y}=0} \bigg(& \prod_{a<b} (y_a - y_b)^{2-2g}  \\
&\times \prod_a \left(N_f (1-y_a) + \lambda y_a\right)^g \frac{(1-y_a)^{-\lambda d_a - (1-g)(N_c-1) -1}}{y_a^{\delta_a +1}} \dd y_a\bigg)\text{.}
\end{aligned}
\end{equation}
This is often a more convenient form. 

The general structure of the quantity in the residue in \eqref{eq:residue1} is
\[
 \underbrace{ \prod_{a>b} V(x_a,x_b)^{\chi(\Sigma)} }_{\substack{\text{integration over off-diagonal} \\ \text{components of gauge group}}} \underbrace{\sum_{\sum k_a = k} \prod_a \underbrace{H(x_a)^g}_{\substack{\text{contribution from}\\ \text{cycles in Jacobian}}} Z(x_a,k_a) \dd x_a }_{\text{Abelian contribution}} \text{,}
\]
where \(V(x_a,x_b) = \left(e^{-x_a} - e^{-x_b} \right)\) comes from the off-diagonal part of the gauge group, \(H(x_a) = (N_f e^{-x_a} + \lambda (1-e^{-x_a})\) comes from integrating over the cycles in the Jacobian (as we saw in \autoref{sec:avm}), and \(Z(x_a,k_a) = e^{(\lambda d_a +(1-g)(N_c-1) +1) x_a}(1-e^{-x_a})^{-(\delta_a +1)}\). This is of the same structure as \cite[Eq. 4.21]{miyakeVMS}, which gives the analogous residue formula for the volume of the vortex moduli space (in our language, the volume is \(\int_\mc{M} \text{ch}(\mc{L})\)), although the precise details differ.

\subsection{The Vandermonde determinant}

The evaluation of the residues \eqref{eq:residue2} requires us to understand the meaning of the Vandermonde-type contribution
\begin{equation}
\label{eq:vandermonde}
\prod_{a<b} (y_a - y_b)^{2-2g} \text{.}
\end{equation}
The meaning is clear for positive powers (that is, for \(g=0\)), but needs thought otherwise, as the residue requires us to understand how it behaves when all the \(y_a \to 0\). 

We argue that the correct series expansion of
\[
(y_a - y_b)^{-2n}
\]
for \(n \in \mathbb{Z}_{\geq 0}\) for our purposes is the symmetric one:
\begin{align*}
(y_a - y_b)^{-2n} &= (-1)^{-n} (y_a - y_b)^{-n}(y_b - y_a)^{-n} \\
	&= (-y_ay_b)^{-n} \left(1 - \frac{y_a}{y_b}\right) \left(1 - \frac{y_a}{y_b}\right) \\
	&= (-y_a y_b)^{-n} \left( 2 - \left( \frac{y_a}{y_b} + \frac{y_b}{y_a} \right) \right)^{-n} \\
	&= (-2y_a y_b)^{-n} \sum_{i=0}^\infty (-1)^i{-n \choose i} \left( \frac{1}{2} \left( \frac{y_a}{y_b} + \frac{y_b}{y_a} \right) \right)^i \text{.}
\end{align*}
If one views this as a power series in the real numbers, it never converges, which is a cost of insisting on symmetry. 

In general, we say that
\begin{align}
(y_a - y_b)^{2 - 2g} &= (y_a - y_b)^2 (y_a - y_b)^{-2g}  \nonumber \\
	&=  (y_a - y_b)^2 (-2y_a y_b)^{-n} \sum_{i=0}^\infty (-1)^i{-n \choose i} \left( \frac{1}{2} \left( \frac{y_a}{y_b} + \frac{y_b}{y_a} \right) \right)^i \text{.} \label{eq:vandermondeseries}
\end{align}

\subsection{The general result}

We can evaluate the general residue \eqref{eq:residue2} using the expansion of the Vandermonde determinant given above:
\[
\prod_{a<b} (y_a - y_b)^{2-2g} = 2^{(N_c-1)g} \sum_{\sum l_a = N_c(N_c-1)} \sum_{c_1, \cdots, c_{N_c} = -\infty}^\infty a_{l_1, \cdots, l_{N_c}} A_{c_1, \cdots, c_{N_c}} \prod_a y_a^{l_a +c_a - g(N_c-1)} 
\]
where \(a_{l_1, \cdots, l_{N_c}}\) is the coefficient of \(\prod_a y_a^{l_a}\) in \(\prod_{a<b} y_a ^{l_a}\) and \(A_{c_1, \cdots, c_{N_c}}\) is the coefficient of \(\prod_a y_a ^{c_a}\) in \(\prod_{a<b}\left(1-1/2(y_a y_b^{-1} + y_b y_a^{-1})\right)^{-g}\).

Evaluating the residue, we find that
\begin{equation}
\label{eq:generalvortexcount}
\begin{aligned}
\chi&(\mc{L}^\lambda) = \frac{(-1)^\sigma 2^{(N_c-1)g}}{N_c!} \sum_{\sum k_a = k} \sum_{\sum l_a = N_c(N_c - 1)} \sum_{c_1, \cdots ,c_{N_c}}  a_{l_1, \cdots, l_{N_c}} A_{c_1, \cdots, c_{N_c}} \\
& \times \prod_a \left( \sum_{j=0}^g {g \choose j} \lambda^j N_f^{g-j} {\lambda (\mc{A} - k_a) + (N_c-1) + N_fk_a +(N_f-1)(1-g) -l_a -c_a - g \choose N_f k_a +(N_f-1)(1-g) - l_a - c_a + g(N_c - 1) -j} \right) \text{,}
\end{aligned}
\end{equation}
which is the general result.

As written, this is not too meaningful (or, at least, not too useful). In what follows, we will attempt to simplify \eqref{eq:generalvortexcount} in special cases.

\section{Simplifications of the index for local vortices}
\label{sec:simp}

\subsection{NonAbelian local vortices on the sphere}

\subsubsection{Warm-up: The case of \(N=2\)}

In this section, we will simplify the form of the index \eqref{eq:generalvortexcount} in the case that \(N_c = N_f \coloneqq N\). We will find ourselves able to simplify the index dramatically when we take \(\lambda = N\).

To warm up, let us consider the case of \(N=2\) and \(g=0\). Results in the general case will follow from similar methods, but it is instructive to begin in this restricted situation as everything is more-or-less directly computationally tractable. 

The relevant residue is
\[
\chi(\mc{L}_{k;2,2}^\lambda) = \frac{(-1)^{k+1}}{2} \sum_{k_1 + k_2 = k} \text{Res}_{x_1 = x_2 =0} (e^{-x_1} - e^{-x_2})^2\frac{e^{( \lambda d_1 +1 ) x_1} e^{(\lambda d_2 + 1) x_2}}{(1-e^{-x_1})^{\delta_1+1} (1-e^{-x_2})^{\delta_2+1}} \dd x_1 \dd x_2 \text{,}
\]
where \(d_a = \mc{A} - k_a\) and \(\delta_a = 2k_a +1\). 

This can be evaluated: it is
\begin{align*}
\frac{(-1)^{k+1}}{2} \sum_{k_1 + k_2 = k} \bigg[ &{\lambda d_1 + \delta_1 -1 \choose \delta_1 -2} {\lambda d_2 + \delta_2 +1 \choose \delta_2 } \\
&-2 {\lambda d_1 + \delta_1  \choose \delta_1 -1} {\lambda d_2 + \delta_2 \choose \delta_2 -1} + {\lambda d_1 + \delta_1 +1\choose \delta_1 } {\lambda d_2 + \delta_2 -1 \choose \delta_2 -2} \bigg]\text{.}
\end{align*}
In general, this appears to be rather complicated.

There is a drastic simplification if \(\lambda =2\), so that \(\lambda = N\), as the \(k_a\) dependence of the binomial coefficients simplifies. In this case, we have
\[
\chi(\mc{L}_{k;2,2}^2) = (-1)^{k+1} \sum_{i=0}^k \left[ {2\mc{A} \choose 2i-1} { 2\mc{A} + 2 \choose 2(k-i)+1} - {2\mc{A}+1 \choose 2i}{2\mc{A} + 1 \choose 2(k-i)}\right]\text{.}
\]
This can be (rather strikingly) simplified further to
\[
\chi(\mc{L}_{k;2,2}^2) \stackrel{!}{=} { 2 \mc{A} \choose k}
\] 
as we show with the following lemmas.

\begin{lemma} 
The identity
\[
\sum_{i=0}^{k} \left[{n \choose 2i-1}{n+2 \choose 2(k-i)+1} - {n+1 \choose 2i}{n+1 \choose 2(k-i)} \right] = \sum_{j=0}^{2k}(-1)^{j+1} {n \choose j}{n \choose 2k-j}
\]	
holds.
\end{lemma}

\begin{proof}
It is possible to derive this result by a liberal application of Pascal's identity for binomial coefficients and a certain rearrangement of the sum, but a much neater proof comes from comparing generating functions on both sides. The sum on the left hand side is the coefficient of \(x^{2k}\) in 
\[
\frac{1}{4} \left[-\left( (1+x)^{n+1} + (1-x)^{n+1} \right)^2 + \left(x^{-1}(1-x)^n - x^{-1}(1+x)^n\right) \left( x(1-x)^{n+2} - x(1-x)^{n+2}\right) \right]\text{,}
\]
where we have used that \(\frac{1}{2} \left( (1+x)^p +(1-x)^p\right) \) isolates the even-power terms of \((1+x)^p\), as well as the analogous result for the odd-power terms.

Expanding this out gives
\[
-\frac{1}{4} ( (1-x) + (1+x))^2(1-x)^n(1+x)^n = -(1-x)^n(1+x)^n \text{.}
\]
The coefficient of \(x^{2k}\) in this final polynomial is
\[
-\sum_{j=0}^{2k} (-1)^j {n \choose j}{n \choose 2k - j}
\]
which gives the claimed result.
\end{proof}

\begin{lemma} 
The identity
\[
\sum_{j=0}^{2k}(-1)^j {n \choose j}{n \choose 2k-j} = (-1)^k{n \choose k}
\]
holds.
\end{lemma}

\begin{proof}
Comparison of the coefficients of \(x^{2k}\) on both sides of the identity
\[
(1-x)^n(1+x)^n = (1-x^2)^n 
\]	
gives the result.
\end{proof}

These lemmas lead us to the result that
\[
\chi(\mc{L}_{k; 2,2}^2) = {2 \mc{A} \choose k} \text{.}
\]

\subsubsection{General \(N\) on the sphere}

When \(g=0\) but \(N\) is allowed to be general, recall that the relevant residue can be evaluated to give
\[
\chi(\mc{L}^\lambda_{k;N,N}) = \frac{(-1)^\sigma}{N!} \sum_{\sum k_a = k} \sum_{\sum l_a = N(N-1)} a_{l_1 \cdots l_{N}}\prod_a {\lambda d_a + \delta_a + N - 1 - l_a \choose \delta_a - l_a}
\]
where \(a_{l_1 \cdots l_{N}}\) is the coefficient of \(y_1^{l_1} \cdots y_N^{l_{N}} \) in the polynomial \(\prod_{a<b} (y_a - y_b)^2\).

This is, again, a complicated sum. Once more, it simplifies if we take \(\lambda = N\). In this case, it becomes
\[
\chi(\mc{L}^N_{k;N, N}) = \frac{(-1)^\sigma}{N!} \sum_{\sum k_a = k} \sum_{\sum l_a = N(N-1)} a_{l_1 \cdots l_N}\prod_a {N\mc{A} + 2(N - 1) - l_a \choose Nk_a + N - 1 - l_a} \text{.}
\]

To understand this sum, we will consider its generating function. The basic result that underlies our computation is the following elementary lemma, which tells us how to pick out the relevant binomial coefficients in our expressions.

\begin{lemma}
\label{lemma:gen1}
For any \(j\), one has
\[
\sum_{i=0}^\infty	{n \choose Ni + j} x^{Ni } =  \frac{x^{-j}}{N}\sum_{a=1}^N r^{-aj}(1+r^ax)^n \text{,}
\]
where \(r\) generates the group of \(N^\text{th}\) roots of unity.
\end{lemma}

\begin{proof}
On the right hand side, we have
\begin{align*}
\frac{x^{j}}{N}\sum_{a=1}^N r^{-aj}(1+r^ax)^n &= \frac{1}{N} \sum_{a=1}^{N} \sum_{l = 0}^\infty {n \choose l} r^{a(l-j)}x^{l-j} \\
	&= \frac{1}{N} \sum_{l=0}^{\infty} {n \choose l} x^{l-j} \sum_{a=1}^N r^{a(l-j)}
\end{align*}
Now, by virtue of the properties of roots of unity, the sum
\[
\sum_{a=1}^N r^{a(l-j)} = r^{l-j}\frac{1-r^{N(l-j)}}{1-r^{l-j}}
\]	
vanishes unless \(l = Ni + j\) for some \(i\), in which case it is \(N\). Our sum thus becomes
\[
\sum_{i=0}^N {n \choose Ni + j} x^{Ni}
\]
as claimed.
\end{proof}

Using this lemma and writing \(n = N\mc{A}\), the Euler characteristic \(\chi(\mc{L}_{k;N,N}^N)\) is, up to the overall factor \((-1)^\sigma N!^{-1}N^{-N}\), the coefficient of \(x^{Nk}\) of the polynomial
\[
P(x) = \sum_{\sum_a l_a = N(N-1)} a_{l_1 \cdots l_N} \prod_{a=1}^n x^{-(N-1)+l_a} \left(\sum_{b=1}^N r^{-(1-l_a) b}(1+r^b x)^{n + 2(N-1)-l_a} \right)
\]
where \(r\) generates the group of \(N^\text{th}\) roots of unity. Because \(\sum l_a = N(N-1)\), we see that \(\prod_a x^{-(N-1) + l_a} = 1\). We then have
\[
P(x) = \sum_{\sum_a l_a = N(N-1)} a_{l_1 \cdots l_N} \prod_a \left(\sum_b r^{-(1-l_a) b}(1+r^b x)^{n + 2(N-1)-l_a} \right) \text{.}
\]

Expanding the product gives
\begin{align*}
P(x) &= \sum_{\sum l_a = N(N-1)} a_{l_1 \cdots l_N} \sum_{b_1 , \cdots, b_N = 1}^N \prod_{a=1}^N r^{-(1-l_a) b_a} (1+r^{b_a} x)^{n+2(N-1) - l_a} \\
	&= \sum_{b_1, \cdots, b_N = 1}^N \sum_{\sum l_a = N(N-1)} a_{l_1 \cdots l_N} \prod_{a=1}^N r^{-(1-l_a )b_a} (1+r^{b_a}x)^{n+2(N-1)-l_a} \text{.}
\end{align*}
The only non-zero contributions to the sum come from those terms where all the \(b_a\) are different. To see this, sum over the terms where \(b_j = \cdots = b_N \eqqcolon b\) for some \(N > j \geq 1\). Write \(a' = 1, \cdots, j-1\) and \(\tilde{a} = j , \cdots, N\). The sum over the restricted configurations with \(b_{\tilde{a}} = b\) is
\[
\sum_{\sum {l_a} = N(N-1)} a_{l_1 \cdots l_N} \sum_{b=1}^N \sum_{b_{a'} = 1}^N\prod_{\tilde{a}=j}^N r^{-(1-l_{\tilde{a}})b}(1+r^bx)^{n + 2(N-1)- l_{\tilde{a}}} \prod_{a'=1}^{j-1}r^{-(1-l_{a'})b_{a'}}(1+r^{b_{a'}}x)^{n + 2(N-1) - l_{a'}}\text{.}
\]
We can now do the product over \(\tilde{a}\), giving
\begin{align*}
\sum_{\sum l_a = N(N-1)} a_{l_1 \cdots l_N} & \sum_b \left( r^{\sum_{\tilde{a}} l_{\tilde{a}} b} (1 +r^b x)^{n + 2(N-1) - \sum_{\tilde{a}} l_{\tilde{a}}} \right) \\
	& \times \sum_{b_{a'}} \prod_{a'=1}^{j-1}r^{-(1-l_{a'})b_{a'}}(1+r^{b_{a'}}x)^{n + 2(N-1) - l_{a'}} \text{.}
\end{align*}

Write \( \tilde{L} = \sum_{\tilde{a}} l_{\tilde{a}} \). Then this becomes
\begin{align*}
\sum_{\tilde{L} = 0}^{N(N-1)} \sum_{\sum l_{\tilde{a}} = \tilde{L}} \sum_{\sum l_{a'} = N(N-1) - \tilde{L}} &a_{l_1, \cdots, l_N} \sum_b r^{\tilde{L}b} (1+r^b x)^{n+2(N-1) - \tilde{L}}\\
	& \times \sum_{b_{a'}} \prod_{a'=1}^{j-1}r^{-(1-l_{a'})b_{a'}}(1+r^{b_{a'}}x)^{n + 2(N-1) - l_{a'}} 
\end{align*}
which vanishes because, for each value of \(\tilde{L}\), the sum
\[
\sum_{\sum_j^N l_{\tilde{a}} = \tilde{L}} a_{l_1 \cdots l_{j-1} l_{j} \cdots l_N} = 0\text{,}
\]
as long as \(N-j > 1\). This follows from the definition of the coefficient \(a_{l_1 \cdots l_N}\) as the coefficient of \(\prod_a y_a^{l_a}\) in \(\prod_{a<b} (y_a - y_b)^2\). 

Returning now to our generating function \(P(x)\), restricting to all of the \(b_a\) different allows us to write
\begin{align*}
\frac{P(x)}{N!} &= \sum_{\sum l_a = N(N-1)} a_{l_1 \cdots l_N}\prod_a r^{-(1-l_a)a} (1+r^a x)^{n + 2(N-1) - l_a} \\
	&= \pm (1\pm x^N)^{n + 2(N-1)} \prod_a (r^{-a} (1+r^ax) )^{-l_a} \text{,}
\end{align*}
where we have used that \(\prod_a (1+r^a x) =  (1\pm x^N)\) where the sign is negative for \(N\) even and positive for \(N\) odd.  We can now use the definition of the \(a_{l_1\cdots l_N}\) to repackage the sum as follows 
\begin{align*}
\frac{P(x)}{N!} &= \pm  (1 \pm x^N)^{n+2(N-1)} \prod_{a<b} \left( \frac{r^a}{(1+r^ax)} - \frac{r^b}{(1+r^bx)} \right)^2 \\
&=\pm (1\pm x^N)^n \prod_{a<b}(r_a - r_b)^2 \\
&= \pm N^N (1 \pm x^N)^{n} \text{.}
\end{align*}

By taking the coefficient of \(x^{Nk}\) and reinserting the relevant factors, we see that the Euler characteristic is
\[
\chi(\mc{L}^N_{k; N, N}) = {N \mc{A} \choose k}
\]
when \(g = 0\).

\subsection{NonAbelian local vortices on general compact surfaces}

\subsubsection{Warming up on general surfaces: The case of \(N=2\)}

To see through the eventual forest of sums and products that arises in general, when \(g\) is general it is again useful to begin in the case of \(N =2\), which is the first nontrivial case. 

The residue that we need to compute is
\[
\frac{(-1)^\sigma 2^{2g}}{2!}  \sum_{k_1+k_2 = k}\text{Res}_{y_1 = y_2 = 0} \left[ (y_1 -y_2)^{2-2g} \frac{(1-y_1)^{-2d_1 - (1-g) - 1}(1-y_2)^{-2d_2-(1-g) - 1}}{y_1^{\delta_1+1}y_2^{\delta_2+1}} \dd y_1 \dd y_2 \right] \text{.}
\]

Our prescription for evaluating this residue is based on the series expansion \eqref{eq:vandermondeseries} for the Vandermonde-type contribution. Plugging this back into the residue and interchanging the sum and the integral gives
\begin{align*}
(-1)^\sigma 2^{g-1}  \sum_{k_1+k_2 = k} &\sum_{i=0}^\infty {-g \choose i} \left(-\frac{1}{2}\right)^i \sum_{j=0}^i {i \choose j} \\
&\times\text{Res}\left[\left(\frac{y_1}{y_2} - 2 + \frac{y_2}{y_1}\right)\frac{(1-y_1)^{-2d_1 - (1-g) - 1} (1-y_2)^{-2d_2- (1-g) - 1}}{y_1^{\delta_1 + (g-1) - 2j + i +1}y_2^{\delta_2 + (g-1)  -i +  2j + 1}} \dd y_1 \dd y_2\right] \text{,}
\end{align*}
where \(d_a = \mc{A} - k_a\) and \(\delta_a = N_f k_a + (N_f-1) (1-g) \).

The residues can be evaluated and the terms grouped, giving
\begin{align*}
(-1)^\sigma 2^{g}  \sum_{k_1+k_2 = k} \sum_{i=0}^\infty {-g \choose i} \left(-\frac{1}{2}\right)^i \sum_{j=0}^i& {i \choose j} \\
\times\bigg[ { 2 \mc{A} + (1-g) -1 - 2j +i \choose 2k_1  -1 -2j +i}&{2\mc{A} + (1-g) + 1 + 2j - i \choose 2k_2  + 1 + 2j - i} - \\
&{ 2 \mc{A} + (1-g)  - 2j +i \choose 2k_1  -2j +i}{2\mc{A} + (1-g)  + 2j - i \choose 2k_2  + 2j - i}  \bigg]\text{.}
\end{align*}
Writing \(S(g,k, 2\mc{A}) \) for this quantity, we will show that the identity
	\[
	S(g,k, n) = {n \choose k}
	\]
holds. 

Introducing the space-saving notation \(n = 2\mc{A}\) and \(l = 2j-i\), the quantity \(S(g,k, n)\) is the coefficient of \(x^{2k}\) of the polynomial
\begin{align*}
(-1)^\sigma 2^{g-2}  \sum_{i=0}^\infty &{-g \choose i} \left(-\frac{1}{2}\right)^i \sum_{j=0}^i {i \choose j} \\
	 \bigg[ &\left( (1+x)^{n +(1-g) -1 - l } + (-1)^{-g-l}(1-x) ^{n +(1-g) -1 - l } \right)\\
	&\times \left((1+x)^{n +(1-g) +1 + l } + (-1)^{-g-l}(1-x) ^{n +(1-g) +1 + l } \right) \\
	&-  \left((1+x)^{n +(1-g) - l } + (-1)^{1-g-l}(1-x) ^{n +(1-g) -l } \right) \\
	&\times \left((1+x)^{n +(1-g) + l } + (-1)^{1-g-l}(1-x) ^{n +(1-g) + l } \right) \bigg] \text{,}
\end{align*}
which simplifies to give
\begin{align*}
(-1)^{\sigma-g} 2^{g-1} \sum_{i=0}^\infty {-g \choose i} &\left(\frac{1}{2}\right)^i \sum_{j=0}^i {i \choose j} (1-x^2)^{n+(1-g)-1-l}\left( (1+x)^{2l+1} + (1-x)^{2l+1}\right) \\
&= (-1)^{\sigma-g} 2^{g} (1-x^2)^{n-g} \sum_{i=0}^\infty {-g \choose i} (-1)^i \left( \frac{x^2 + 1}{x^2-1} \right)^i \\
&= (-1)^{\sigma - g} 2^{g} (1-x^2)^{n-g} \left(\frac{1}{1- \frac{x^2+1}{x^2-1}} \right)^g \\
&= (-1)^k(1-x^2)^n \text{.}
\end{align*}
The coefficient of \(x^{2k}\) in \((1-x^2)^n\) is \( (-1)^k{n \choose k}\), so we have shown that 
\[
S(g,k,n) = {n \choose k}
\]
as claimed.

Modulo technical questions regarding the mathematical validity of the residue prescription we give, we have demonstrated that
\[
\chi(\mc{L}^2_{g;k;2,2}) = {2\mc{A} \choose k}
\]
for all \(g\), \(k\) and \(\mc{A}\). 

\subsubsection{General \(N\) on general surfaces}

We now turn to the most general case of local vortices in \(U(N)\) gauge theory with \(N\) flavours on arbitrary compact surfaces.

Our results for general \(N_c\) and \(g=0\) and for \(N_c=2\) and general \(g\) lead us to suspect that
\[
\chi(\mc{L}_{g;k;N,N}^N) \stackrel{?}{=} {N \mc{A} \choose k}
\]
for all \(N\) and \(g\). 

We again simplify by setting \(\lambda =N_f\). The residue we need to compute is then
\[
\frac{(-1)^\sigma N_f^{N_c g}}{N_c!} \sum_{\sum k_a = k} \text{Res}_{\bsym{y}=0} \left[ \prod_{a<b} (y_a - y_b)^{2-2g} \prod_a \frac{(1-y_a)^{-N_f d_a - (N_c-1) (1-g)-1}}{y_a^{\delta_a + 1}} \dd y_a\right] \text{,}
\]
where, as always, \(d_a = \mc{A} - k_a\), \(\delta_a = N_f k_a + (N_f-1) (1-g) \), and \(a,b = 1, \cdots N_c\).

Our prescription for computing this residue is based on the expansion of \eqref{eq:vandermondeseries}, which gives
\begin{align*}
\prod_{a<b} (y_a - y_b)^{2-2g} &= \prod_{a<b} (y_a - y_b)^2 \prod_{a<b} (y_a - y_b)^{-2g} \\
	&= 2^{(N_c-1)g}\prod_{a<b} (y_a - y_b)^2 \prod_a y_a^{-g(N_c - 1)} \prod_{a<b} \left(1 - \frac{1}{2}\left(\frac{y_a}{y_b} + \frac{y_b}{y_a}\right)\right)^{-g} \\
	&= 2^{(N_c-1)g}\prod_{a<b} (y_a - y_b)^2 \prod_a y_a^{-g(N_c - 1)} \prod_{a<b} \sum_{i_{a,b} = 0}^\infty {-g \choose i_{a,b}} \left(-\frac{1}{2}\right)^{i_{a,b}} \left(\frac{y_a}{y_b} + \frac{y_b}{y_a} \right)^{i_{a,b}} \\
	&= 2^{(N_c-1)g} \prod_{a<b} (y_a - y_b)^2 \prod_a y_a^{-g(N_c - 1)} \prod_{a<b} \sum_{i_{a,b} = 0}^\infty {-g \choose i_{a,b}} \left(-\frac{1}{2}\right)^{i_{a,b}} \\
	& \quad\quad\quad\quad \times \left( \sum_{j_{a,b}=0}^{i_{a,b}} {i_{a,b} \choose j_{a,b}} y_a^{2j_{a,b} - i_{a,b}} y_b^{-2j_{a,b} + i_{a,b}} \right) \\
	&=2^{(N_c -1)g} \sum_{\sum l_a = N_c(N_c-1)} a_{l_1 
	\cdots l_{N_c}} \prod_a y_a^{l_a - g(N_c-1)} \prod_{a<b} \sum_{i_{a,b} = 0}^\infty {-g \choose i_{a,b}} \left(-\frac{1}{2}\right)^{i_{a,b}} \\
	& \quad\quad\quad\quad \times \left( \sum_{j_{a,b}=0}^{i_{a,b}} {i_{a,b} \choose j_{a,b}} y_a^{2j_{a,b} - i_{a,b}} y_b^{-2j_{a,b} + i_{a,b}} \right) \\
	&=2^{(N_c-1)g}\sum_{\sum l_a = N_c(N_c-1)} a_{l_1, \cdots, l_{N_c}} \sum_{c_1, \cdots, c_{N_c} = -\infty}^\infty  A_{c_1, \cdots ,c_{N_c}} \prod_a y_a^{l_a + c_a -g(N_c-1)}
\end{align*}
where the \(a_{l_1\cdots l_N}\) is the coefficient of \(\prod_a y_a^{l_a}\) in \(\prod_{a<b}(y_a - y_b)^2\) as before and \(A_{c_1, \cdots, c_{N_c}}\) is the coefficient of \(\prod_a y_a^{c_a}\) in \(\prod_{a<b} \sum_{i_{a,b}} {-g \choose i_{i,b}} (-1/2 (y_a y_b^{-1} + y_b y_a^{-1}))^{1/2} \).

The idea is to unwind this expansion, compute the residue, pass to the generating function, and then wind the expansion back up to obtain a neat form. 

Substituting this into the residue gives, up to the overall prefactor \((-1)^\sigma 2^{(N_c-1)g} N_f^{N_c g} N_c!^{-1}\),
\[
\sum_{\sum k_a = k} \sum_{\sum l_a = N_c(N_c-1)} \sum_{c_1, \cdots, c_{N_c}} a_{l_1, \cdots, l_{N_c}} A_{c_1, \cdots, c_{N_c}} \prod_a {N_f d_a + (N_c-1) + \delta_a - l_a - c_a \choose \delta_a - l_a - c_a + g(N_c-1)}
\]
where, as always, \(d_a = \mc{A} - k_a \) and \(\delta_a = N_f k_a + (N_f - 1)(1-g) \). This constitutes the general result for the Euler characteristic \(\chi(\mc{L}_{k; N_c, N_f} ^{N_f})\). 

We now set \(N_c = N_f \eqqcolon N\). In this case, the residue becomes
\begin{align*}
\sum_{\sum k_a = k} &\sum_{\sum l_a = N(N-1)} \sum_{c_1, \cdots, c_{N}} a_{l_1, \cdots, l_{N}} A_{c_1, \cdots, c_{N}} \prod_a {N \mc{A}  + (N-1)(2-g) - l_a - c_a  \choose N k_a + (N-1) - l_a - c_a }\text{.}
\end{align*}
Using the \cref{lemma:gen1} and writing \(n = N \mc{A}\), we find that these numbers are the coefficients of \(x^{Nk}\) of the polynomial
\begin{align*}
P_g(x) =  \sum_{\sum l_a = N(N-1)} &\sum_{c_1, \cdots, c_{N}} a_{l_1, \cdots, l_{N}} A_{c_1, \cdots, c_{N}} \\
&\times \prod_a x^{-(N-1)+ l_a + c_a} \left( \sum_{b=1}^N r^{-(1-l_a -c_a)b} ( 1+ r^b x)^{n + (N-1)(2-g)-l_a-c_a}  \right)
\end{align*}
where \(r\) generates the \(N^\text{th}\) roots of unity, as before, and we have multiplied through by \(N^N\).

The method now is similar to the genus zero case. That \(\sum l_a = N(N-1)\) gives \( \prod_a x^{-(N-1)+l_a} = 1\). We then expand the product to give
\begin{align*}
P_g(x) = \sum_{\sum l_a = N(N-1)} &\sum_{c_1, \cdots, c_{N}} a_{l_1, \cdots, l_{N}} A_{c_1, \cdots, c_{N}} \\
&\times \sum_{b_1, \cdots, b_N = 1}^N \prod_a x^{ c_a }  r^{-(1-l_a -c_a)b_a} ( 1+ r^{b_a} x)^{n + (N-1)(2-g)-l_a-c_a} \text{.}
\end{align*}
For exactly the same reasons as before the only nonzero contributions come when all the \(b_a\) are distinct, at least after regularisation. Thus
\begin{align*}
\frac{P_g(x)}{N!}  &= \sum_{\sum l_a = N(N-1)} \sum_{c_1, \cdots, c_{N}} a_{l_1, \cdots, l_{N}} A_{c_1, \cdots, c_{N}}  \prod_a x^{ c_a }  r^{-(1-l_a -c_a)a} ( 1+ r^a x)^{n + (N-1)(2-g)-l_a-c_a} \\
 &= (-1)^{N+1} (1\pm x^N)^{n+(N-1)(2-g)} \sum_{l_a} \sum_{c_a} a_{l_1, \cdots l_{N}} A_{c_1, \cdots, c_N}  \\
 & \quad\quad\quad\quad\quad\quad\quad\quad\quad\quad\quad \times \prod_a (x^{-1} r^{-a}(1+r^ax))^{-c_a}(r^{-a}(1+r^ax))^{-l_a} \text{.}
\end{align*}

We now wind up the series, using the definition of the coefficients \(a_{l_1,\cdots, l_N}\) and \(A_{c_1, \cdots, c_N}\), which leads to the polynomial
\begin{align*}
\frac{P_g(x)}{N!} = (-1)^{N+1} (1\pm x^N)^{n+(N-1)(2-g)} &\prod_{a<b}\left(\frac{r^a}{(1+r^ax)} - \frac{r^b}{(1+r^bx)} \right)^2 \\
	& \times \prod_{a<b} \left(1 - \frac{1}{2} \left(\frac{r^a(1+r^bx)}{r^b (1+r^a x)} + \frac{r^b(1+r^ax)}{r^a(1+r^bx)} \right)\right)^{-g}\text{.}
\end{align*}
This is
\begin{align*}
\frac{2^{(N-1)g}}{N!} P_g(x) &= \pm (1\pm x^N)^{n+(N-1)(2-2g)} \prod_{a<b}\left(\frac{r^a}{(1+r^ax)} - \frac{r^b}{(1+r^bx)} \right)^{2-2g} \\
	&= \pm N^{N(1-g)}(1\pm x^N)^{n} \text{.}
\end{align*}

Now, \(P_g(x)\) was defined to so that the coefficient of \(x^{Nk}\) in 
\[
(-1)^\sigma N^{N(g-1)}N!^{-1} 2^{(N-1)g} P_g(x)
\]
is the quantity we are after. Collecting all of the various prefactors that we have discarded and comparing coefficients, this gives the final result
\begin{equation}
\label{eq:localcount}
\chi(\mc{L}_{g; k; N,N}^N ) = {N \mc{A} \choose k} \text{,}
\end{equation}
as expected.

\subsection{A heuristic argument for the local vortex count}
\label{subsec:heuristic}

The computation that we have done is interesting in that it is rather involved but has, in the local case (at least), a very simple output. One simple, but heuristic, way to see that the local vortex count might be
\[
{N \mc{A} \choose k}
\] 
follows by considering symmetry arguments based on colour-flavour locking, as used in \cite{hananyVIB,tongTASI} to understand the internal moduli of vortices. Let us summarise the construction. Let \((A^*, \phi^*)\) be a solution to the local \emph{Abelian} vortex equations. Then one can embed this Abelian solution into a solution \((A, \phi)\) to the local nonAbelian vortex equations as
\[
A = \left(
\begin{matrix}
A^* && & \\
&0 && \\
&& \ddots & \\
& && 0	
\end{matrix}
\right) \text{, } \quad
\phi = \left(
\begin{matrix}
\phi^* &&& \\
&\sqrt{\tau} && \\
&& \ddots & \\
&&& \sqrt{\tau}	
\end{matrix}
\right) \text{.}
\]
Of course, there is nothing special about the top left entry: a local nonAbelian vortex can be given by putting a local Abelian vortex in any colour-flavour `slot'. 

One could then build a charge \(k\) nonAbelian vortex by distributing \(k\) Abelian vortices into the \(N\) `slots'. Given that the number of charge \(k_i\) local Abelian vortices is
\[
{\mc{A} \choose k_i} \text{,}
\]
the number of ways to distribute \(k\) Abelian vortices among the \(N\) slots is
\[
\sum_{k_1 + \cdots + k_N = k} {\mc{A} \choose k_1} \cdots {\mc{A} \choose k_N} = {N \mc{A} \choose k} \text{.}
\]
This is in agreement with the result \eqref{eq:localcount}.

Note that, while the result of the localisation computation also came as a sum over partitions of \(k\), the form that the summands take is very different to this and there are cancellations between the contributions from the various partitions.

It is quite remarkable that objects as complicated as vortices might be amenable to such simple symmetry arguments (at least in special cases). It would be interesting to see how to make the above argument rigorous. The results of Baptista in \cite{baptistaNAV} show that aspects of these symmetry-based arguments can be made rigorous in some generality when applied to understanding the geometry of the vortex moduli space.

\section{Local duality}
\label{sec:localduality}

\subsection{What are local vortices?}

In our discussion of duality we will restrict ourselves to two special cases where our results take a particularly simple form: the local case of \(N_f = N_c = \lambda  \) and the Abelian semilocal case of \(N_c =1\), \(N_f = \lambda\). (Recall that working at Chern--Simons level \(\lambda = N_f\) simplifies our formulae significantly, especially on surfaces of genus greater than one.)

We begin here in the local case, where we conjecture that quantum vortices `are' fermions in a background flux, generalising the results of \cite{erikssonKQ}.  

More precisely, we will find that our vortex quantum mechanics seems to agree with the theory of \(N \coloneqq N_f = N_c = \lambda\) flavours of (spin half) fermion coupled to a background gauge field for the flavour symmetry. Matching of these theories requires the matching of choices made in defining the theory. We will see that both theories depend on the choice of a holomorphic vector bundle.

\subsection{The Fermi theory}

The basic Fermi theory dual to the theory of local \(U(N)\) vortices at level \(N\) takes the following simple form, based on the Abelian example of \cite{erikssonKQ}. 

Let \(V \to \Sigma\) be a (flavour) complex vector bundle of rank \(N\) and degree \(c_1(V)[\Sigma] = N\mc{A} \), which should be an integer. The field of the static theory is a section
\[
\psi = (\psi_1, \cdots, \psi_N) \in \Gamma( V \otimes K_\Sigma^{1/2})
\]
where \(K_\Sigma^{1/2}\) is a spin bundle on \(\Sigma\). Each of the \(\psi_i\) is interpreted as an individual flavour. Note, however, that for general \(N\mc{A}\) there is no splitting 
\[
V \stackrel{?}{=} L^{\oplus N}
\]
for a line bundle \(L\), so the individual flavours do not really have an independent existence. Of course, one can make sense of them locally. If \(\mc{A}\) is an integer, so that \(N\mc{A} \in N\mathbb{Z}\), then \(L\) does exist, and one can make sense of the individual flavours.

We give the flavour bundle \(V \otimes K_\Sigma^{1/2}\) a backgound connection \(B\). The total flux of \(B\) through \(\Sigma\) is fixed by the topology of the bundle to be \({N} (\mc{A} - 1 + g) \). The connection \(B\) induces a holomorphic structure on \(V \otimes K_\Sigma^{1/2}\) via the Dolbeault operator \(\db_B \).

We extend the bundle \(V \otimes K_\Sigma^{1/2}\) to \(\Sigma \times S^1\) in the trivial way, pulling it back along the projection to \(\Sigma\). We then think of \(\psi\) as being a function of a periodic time component \(t\). We should also introduce a time component \(B_t\) to the background gauge field \(B\) so as to preserve gauge invariance. However, because the topology of the bundle is trivial along the time direction, we can make the gauge choice \(B_t = 0\). This is sometimes called the temporal gauge. As \(B\) is a background gauge field, this is completely harmless.

In this gauge, the Lagrangian is the Pauli Lagrangian
\[
\psi^\dagger( \ii \partial_t + \Delta_B - *F(B))\psi + \text{c.c.}
\]
where \(\Delta_B\) is the usual \(B\)-coupled Laplacian and the combination \(\Delta_B - *F(B)\) gives the Dolbeault Laplacian \(\Delta_{\db_B} = \db_B^\dagger \db_B\).

Note that if \(\mc{A} \in \mathbb{Z}\) and \(B\) diagonalises, this is the Lagrangian for \(N\) fermion flavours in the background of a (genuine) magnetic field of total flux \(\mc{A}\). 

At low temperature, the theory becomes fermionic Hamiltonian mechanics on the space of zero modes
\[
\db_B \psi = 0 \text{,}
\]
which is the zeroth vector bundle cohomology \(H^0(V \otimes K_\Sigma^{1/2})\).
The Riemann--Roch theorem tells us that the `expected' dimension of this space is
\begin{align*}
h^0(V\otimes K_\Sigma^{1/2}) - h^1(V \otimes K_\Sigma^{1/2}) &=  N(\mc{A} - 1 +g) + N(1-g) \\
	&= N\mc{A} \text{.} 
\end{align*}

When is \(h^1(V \otimes K_\Sigma^{1/2}) = 0\)? Suppose that \(\mc{A} \in \mathbb{Z}\) and that \(V\) splits as \(L^{\oplus N}\) as a holomorphic bundle where \(L\) is a holomorphic line bundle of degree \(\mc{A}\). In that case, the first cohomology of \(V\otimes K_\Sigma^{1/2} = (L\otimes K_\Sigma^{1/2})^{\oplus N}\) vanishes whenever that of \(L \otimes K_\Sigma^{1/2}\) vanishes. By positivity, this happens for \(\mc{A} > g- 1\). It seems likely that in the general case, there will be similar vanishing results for \(N\mc{A} \gg Ng\).

Suppose we are in this regime. Then we have a nonrelativistic quantum mechanics on the odd vector space
\[
\Pi \mathbb{C}^{N\mc{A}} \text{,}
\]
where \(\Pi\) denotes parity reversal.
The geometric quantisation prescription tells us that quantum states should be holomorphic functions on this space. Such functions are analytic functions of the holomorphic coordinates \((z_1, \cdots, z_{N\mc{A}})\). Because the \(z_i\) are anticommuting, such a function is the linear combination of monomials of the form
\[
z_{i_1} \cdots z_{i_k}
\]
where \(0 \leq k \leq N\mc{A}\). The total Hilbert space is therefore the exterior algebra
\[
\mc{H}^\bullet_\text{Fermi} = \Lambda^\bullet \mathbb{C}^{N\mc{A}}
\]
with total dimension \(2^{N\mc{A}}\).

There is an overall \(U(1)\) action on the space of zero modes, acting by \(z_i \mapsto e^{\ii \theta} z_i\) for each \(i\). This grades the Hilbert space by particle number. The \(k\)-particle Hilbert space is
\[
\mc{H}^k_\text{Fermi}  = \Lambda^k\mathbb{C}^{N\mc{A}} \text{,}
\]
which has dimension \({N\mc{A} \choose k}\). This is the expected dimension of the \(k\)-vortex Hilbert space.

This is the basic result underpinning the duality \eqref{eq:duality1} which we gave at the start of this paper:
\[
Z_{\text{SCS},N} \,\, \leftrightarrow \,\, Z_{\text{Fermi},N} \text{.}
\]
We now understand the theories on both sides a little better. When \(\mc{A} \in \mathbb{Z}\), so that the Fermi theory becomes the theory of \(N\) fermions in a magnetic field, this exactly realises the low energy vortex theory as the theory of the lowest Landau level for a system of \(N\) Fermi flavours in a magnetic field of flux \(\mc{A}\).
Really though, for general \(\mc{A}\), this should be interpreted as a Landau level of charge \(1/N\) objects (or, more accurately, a nonAbelian Landau level).

There is a subtlety in the matching of the global symmetries in the Fermi-vortex duality above. In the Fermi theory, the global symmetry is 
\[
U(N) = \frac{U(1) \times SU(N)}{\mathbb{Z}_N} \text{.}
\]
The \(U(1)\) factor allows us to count the number of particles, and the \(SU(N)\) factor rotates the flavours among themselves. The discrete quotient allows us to couple to background gauge fields for the flavour symmetry which would not be viable as \(U(1) \times SU(N)\) gauge fields. This is the reason that the `effective Abelian flux' \(\mc{A}\) need not be an integer, as long as \(N\mc{A}\) is.

Na\"ively, the vortex theory has a topological global symmetry
\[
U(1)_\text{top}
\]
which allows us to count the number of vortices, and a flavour symmetry
\[
SU(N)_\text{flavour}.
\]
In the Higgs phase, the flavour symmetry is locked to the colour group as
\[
SU(N)_\text{diag} \subset U(N)_\text{colour}\times SU(N)_\text{flavour} \text{.}
\]
which rotates the possible vortex charges among themselves at the same time as the scalar field flavours.

How do these global symmetries fit together? One way to understand this is to start in the vortex theory and couple the scalar fields to a background gauge field for a \(U(N)\) flavour symmetry. This may look illegal, because the overall \(U(1)\) is gauged, but there is nothing stopping one from thinking of the scalar field as taking values in
\[
E \otimes V \to \Sigma
\]
where \(E\) is the colour bundle on which there is a dynamical gauge field and \(V\) is a flavour bundle which could have nonzero degree, \(d\). This nonzero degree shifts the overall Chern class of the bundle \(E\otimes V\) by \(N d\) and so is equivalent to increasing the degree of the colour bundle by \(d\). This means shifting the vortex number \(k \mapsto k + d\). This slightly roundabout argument implies that the global symmetries fit together as
\[
\frac{U(1)_\text{top} \times SU(N)_\text{flavour}}{\mathbb{Z}_N} = U(N) \text{,}
\]
so that the global symmetries match across the duality.

The equality of indices that we have demonstrated (as well as the matching of global symmetries and various parameters) is good evidence for the duality, but does not constitute a proof of the duality. To prove the duality one needs to show more.

In the Abelian case \(N=1\), Eriksson--Rom\~ao showed in \cite{erikssonKQ} that the \(k\)-vortex Hilbert space was
\[
\mc{H}^k_{N=1} \cong \Lambda^k H^0(\Sigma, L \otimes K_\Sigma^{1/2})
\]
where \(L\) is a holomorphic line bundle of degree \(\mc{A}\) as above. This result holds for all \(\mc{A} \in \mathbb{Z}\) and does not rely on vanishing theorems (in particular, the dimension of the Hilbert space may be larger than the expected dimension for small \(\mc{A}\)). Note that the vortex Hilbert space is defined with respect to a line bundle on \(\Sigma\) (corresponding to a choice of normalisation for a Poincar\'e bundle which allows for the natural construction of a quantum line bundle on the vortex moduli space) which must be taken to be \(L \otimes K_\Sigma^{1/2}\) for the above isomorphism to be natural and even to exist in general.

This leads us to conjecture in the nonAbelian case that the \(k\)-vortex Hilbert space is
\begin{equation}
\label{eq:localdualityconj}
\mc{H}^k_{N} \stackrel{?}{\cong} \Lambda^k H^0(\Sigma, V \otimes K_\Sigma^{1/2})
\end{equation}
where \(V\) is a (semi)stable holomorphic bundle of rank \(N\) and degree \(N\mc{A}\) as above, and the vortex Hilbert space is defined with respect to the bundle \(V \otimes K_\Sigma^{1/2}\). Proving this conjecture would essentially constitute a proof of the Fermi-vortex duality as a genuine duality of non-supersymmetric field theories.

The conjecture would also imply that the true dimension of the Hilbert space obeys the inequality
\[
\text{dim} (\mc{H}_N^k) \geq {N\mc{A} \choose k}\text{,}
\]
which follows because 
\begin{align*}
h^0(\Sigma , V \otimes K_\Sigma^{1/2}) &= N\mc{A} + h^1(\Sigma, V \otimes K_\Sigma^{1/2}) \\
	&\geq N\mc{A}\text{.}
\end{align*}

\section{Semilocal duality}
\label{sec:slduality}

\subsection{What are semilocal vortices?}

The outcome of the vortex count at \(N_c = 1\) is (see \autoref{theorem:semilocalab})
\[
\chi(\mc{L}_{k;1,N_f}^\lambda ) = \sum_{j=0}^g \lambda^j N_f^{g-j} {g \choose j} {\lambda (\mc{A} - k) + \delta -g \choose \delta-j} 
\]
where \(\delta = N_f k + (N_f-1)(1-g)\).

The simplest case is \(\lambda = N_f\), in which case the count is
\begin{equation}
\label{eq:slcount}	
N_f^g{N_f \mc{A} + (N_f -1)(1-g) \choose N_f k + (N_f-1)(1-g) } = N_f^g {N_f \mc{A} + (N_f-1)(1-g) \choose N_f(\mc{A} - k)} \text{.}
\end{equation}
The first thing to note here is the asymmetry between vortices and vortex-holes for \(g\neq 1\): it appears that vortex-holes are more natural from the perspective of state counting, reflecting the strange-looking selection rule \(N_fk \geq (N_f-1)(g-1) \) for semilocal Abelian vortices (this is the same as the condition that the expected dimension of the vortex moduli space be nonnegative). The second thing to note is that the count of vortex-holes is reminiscent of the counting of states consisting of \(N_f\) anticommuting objects (that is, fermions). The third thing to note is the appearance of the factor \(N_f^g\), which is familiar from state counting in Abelian Chern--Simons and in theories of anyons.

In this case, we conjecture that the vortex theory is dual to a \(U(1)^{N_f}/U(1)\) theory coupled to \(N_f\) fermion flavours, each attached to flux and spin. We give a description of this theory.

More generally, away from \(\lambda = N_f\), topological effects become more important, exemplified by the fact that the index generally comes as a sum over \(0,\cdots,g\). We do not give a general field-theoretic description in this case, but conjecture that the vortex theory is dual to a theory of anyons with fractional statistics \(\exp(\ii\pi \lambda /N_f)\).

\subsection{Global symmetries}

To understand what a potential dual theory looks like, we consider the symmetries of the semilocal theory. There is the flavour symmetry
\[
SU(N_f)
\]
which is broken to \(S(U(N_c)_\text{diag} \times U(N_f-N_c))\) in the Higgs phase, and the vortex-counting topological symmetry
\[
U(1)_\text{top} \text{.}
\]

To see how the \(U(1)\) and \(SU(N)\) symmetries fit together, we again think of the scalar fields as taking values in a bundle
\[
E \otimes V \to \Sigma
\]
where \(E\) is the colour bundle and \(V\) is the flavour bundle. Shifting the degree of the flavour bundle by \(d\) is formally equivalent to a shift of \(E\) by \(\frac{N_f}{N_c} d\). 

In general, the fractional nature of the shift means that there is something to think about here\footnote{In a sense, the fact that the shift simplifies for \(N_f = N_c\) or for \(N_c = 1\) is the reason why these two cases are simple enough for us to study in detail.}, but in the Abelian case this means that the topological symmetry binds to the original flavour symmetry as
\begin{equation}
\label{eq:semilocalsym}
U(1)_\text{top} \times SU(N_f)\text{,}
\end{equation}
which is the global symmetry we should aim for.

\subsection{Flux attachment, spin attachment, and Chern--Simons theory}

Consider the theory of a spin \(1/2\) fermion \(\psi\) coupled to a \(U(1)\) gauge potential \(a\) at Chern--Simons level \(n\), with Lagrangian
\[
\frac{n}{2\pi}\, * \, a \wedge \dd a - \frac{1}{2} \left ( \psi^\dagger ( \ii D_t + \Delta_{\db_{a\otimes B}} ) \psi + \text{c.c.} \right) \text{,}
\]
where \(B\) is a background Abelian connection, including the spin connection (often, one incorporates \(B\) into the dyanamical connection \(a\) by including a BF coupling of the form \(a \wedge \dd B\)). 

The equation of motion for \(a_t\) is the Gauss law constraint
\begin{equation}
\label{eq:fluxattach}
\frac{n}{2\pi} F(a)_\Sigma =  |\psi |^2 \omega_\Sigma\text{.}
\end{equation}
The object on the right is the \(\psi\) particle density on \(\Sigma\). This equation requires the particle number, which is the integral of this particle density, to be quantised in units of \(n\). It also tells us that for a given configuration of \(\psi\) particles, the moduli space of classical gauge fields on \(\Sigma\) is the Jacobian of \(\Sigma\) with symplectic form scaled by \(n\). It is well-known that quantising this gives \(n^g\) states associated to the gauge field \(a\).

On the other hand, the low-temperature equation of motion for \(\psi\) on a compact surface is the lowest Landau level equation
\[
\db_{a\otimes B} \psi = 0 \text{.}
\]
The Riemann--Roch theorem tells us that, if the background magnetic flux is sufficiently strong, the number of possible states for the \(\psi\) particle is
\[
\mc{B} + m
\]
where \(\mc{B}\) is the integrated magnetic flux, and \(m\) is the number of \(\psi\) particles divided by \(n\). The genus-dependent terms cancel because \(\psi\) has spin \(1/2\).

Because \(\psi\) is a fermion, the occupation number of each state is one or zero. This means that, for a given compatible configuration of the gauge field, the number of \(\psi\) states is
\[
{\mc{B} + m \choose nm}\text{.}
\]
This does not take into account the number of states associated to the gauge field, which to a first approximation, gives an additional factor of \(n^g\). In general, a more detailed calculation is necessary to account for potential effects associated to nontrivial topology of the moduli space. 

Proceeding for now, if we set \(n=1\) then \(n^g = 1\) and the (approximate) number of quantum states in the full theory is
\begin{equation}
\label{eq:boson1}
{\mc{B} + m \choose m}\text{.}
\end{equation}
This is the count of states in a Landau level of \emph{bosons}.

This carries the essence of \emph{flux attachment} \cite{wilczekQM}. The equation \eqref{eq:fluxattach} `attaches' \(1/n\) units of flux to each \(\psi\) particle. Each \(\psi\) particle is also electrically charged and the flux attachment turns it into a flux-charge composite. By considering the Aharanov--Bohm phase associated to a rotation of this composite, one can see that the statistical exchange phase of \(\psi\) is shifted by \(\exp(\ii \pi n^{-1})\). As we started with a fermion, the new statistical phase is
\[
\exp( \ii \pi ( 1/n  - 1) ) \text{.}
\]
When \(n =1\), the composite is a boson, which agrees with what we saw with the state count in the lowest Landau level.

There is a subtlety here. The state count \eqref{eq:boson1} is \emph{not} the count of states for spin 0 bosons in the lowest Landau level in general, but rather the count of states for spin \(1/2\) bosons. Indeed, the number of states for a system of \(q\) spin 0 bosons in the lowest Landau level instead takes the form
\[
{\mc{B} + q + (1-g) \choose q} \text{.}
\]
Of course, one can absorb \(1-g\) into \(\mc{B}\), shifting the total effective magnetic flux. This is a little unsatisfactory from a geometrical perspective. It is a little more natural to think of flux attachment as changing the statistics, but not the geometrical spin. 

If we want to really change the spin of the fermion, we should introduce a BF coupling between the dynamical connection \(a\) and the background spin connection \(\Gamma\) on \(K_\Sigma^{1/2}\) so as to `eat' the spin degrees of freedom. This term takes the general form
\[
  \frac{p}{2\pi} \, * \, a \wedge \dd \Gamma \text{,}
\]
for some level \(p\). Note that \(\Gamma\) is not dynamical: it is a background connection. The new flux attachment equation is
\[
n F(a)_\Sigma + p F(\Gamma)_\Sigma = 2\pi |\psi|^2 \omega_\Sigma \text{.}
\]
Integrating this, we see that the allowed number of \(\psi\) particles now takes the form
\[
q = nm - p(1-g)
\]
for \(m \in \mathbb{Z}\). We have used that \(\Gamma\) is a connection on \(K_\Sigma^{1/2}\) and so has degree \(g-1\).

Setting \(n = p = 1\), the count of states with \(q = m - 1 + g\) of the fermionic \(\psi\) particles is then
\[
{\mc{B} + q +(1-g) \choose q}
\]
which is the count of states for \(q\) spinless bosons, as desired. 

\subsection{The dual of a semilocal vortex}

We now turn these ideas towards the Abelian semilocal vortex theory. For simplicity, consider the case of \(\lambda = N_f\) so that the vortex count is
\[
N_f^g { N_f \mc{A} + (N_f -1)(1-g) \choose N_f (\mc{A} - k)}\text{.}
\]
When \(N_f = 1\), this is simply a theory of fermions in a magnetic field. When \(N_f > 1\), we must do something a bit more clever. 

Although, for consistency, we will eventually start with a theory of fermions, it is actually easier to write down a bosonic dual in the general case. Take a theory of \(N_f\) spinless bosons \(\Phi = (\Phi_1, \cdots, \Phi_{N_f})\) coupled to a background \(U(1) \times SU(N_f)\) connection \(b\) (this differs from our usual \(B\) as it does not contain the spin connection, because our bosons are spinless) and a dynamical \(U(1)\) gauge potential \(a\). We introduce the Chern--Simons coupling
\[
\frac{N_f}{2\pi} \, * \, a \wedge \dd a - \frac{1}{2\pi} \, * \, a \wedge \dd \Gamma
\]
where \(\Gamma\) is the non-dynamical spin connection on \(K_\Sigma^{1/2}\), and the usual matter terms
\[
 \frac{1}{2} \left( \Phi^\dagger ( \ii D_t + \Delta_{\db_{a\otimes b} }) \Phi + \text{c.c.} \right)
\]
into the Lagrangian. 

The Gauss law is 
\[
- N_f F(a)_\Sigma + F(\Gamma)_\Sigma = 2\pi | \Phi|^2 \omega_\Sigma \text{,}
\]
which restricts the number of \(\Phi\) particles to take the form
\[
q = - N_f m + (1-g)\text{,}
\]
for \(m \in \mathbb{Z}\). 

At low temperatures, the equation of motion for \(\Phi\) is
\[
\db_{a\otimes b} \Phi = 0
\]
which has a space of solutions of expected dimension
\[
N_f \mc{B} + N_f m + N_f(1-g) = N_f \mc{B} - q + (N_f-1)(1-g)
\]
where \(\mc{B}\) is the total flux associated to the background magnetic field \(b\). Because \(\Phi\) is bosonic, the counting of \(q\) particle states is, assuming that the count of states factors simply into gauge and particle counts,
\[
N_f^g{N_f \mc{B} + (N_f-1)(1-g) \choose -N_f m + (1-g)}\text{.}
\]
If we identify the vortex number \(k\) with \(-m + 1 -g\) and \(\mc{A}\) with \(\mc{B}\), this is the index for semilocal vortices at \(\lambda = N_f\). This is a version of bosonic particle-vortex duality.

It would be nice to give this a fermionic description. To do this, we should realise each of the \(\Phi_i\) as a fermion attached to a unit of flux and spin. 

To do this, we start with \(N_f\) fermion flavours \(\psi_1, \cdots, \psi_{N_f}\). We gauge the \(U(1)\) phase of each fermion flavour individually, writing \(\alpha_i\) for the gauge field associated to rotations of \(\psi_i\). We introduce an extra dynamical gauge field \(a\), which is associated to simultaneous rotations of all of the \(\psi_i\). This might look a bit strange, as we have already gauged that \(U(1)\), but one can think of it as a shift of each of the \(\alpha_i\). This means that the overall gauge group is \(U(1)^{N_f} / U(1)\).

We consider the Chern--Simons terms
\[
\frac{N_f}{2\pi} a \wedge \dd a - \frac{1}{2\pi } a \wedge \dd \Gamma - \sum_{i=1}^{N_f} \left( \frac{1}{2\pi} \alpha_i \wedge \dd \alpha_i + \frac{1}{2\pi} \alpha_i \wedge \dd \Gamma \right) \text{,}
\]
in addition to the usual matter terms for \(\psi\).

The Gauss laws for the various gauge fields impose the equations
\begin{align*}
F(\alpha_i) + F(\Gamma) &= 2\pi \psi_i^\dagger \psi_i \text{ (no summation)}\\
-N_f F(a) + F(\Gamma) &= 2\pi \sum_i \psi_i^\dagger \psi_i \text{.}
\end{align*}
These equations imply that
\[
\sum_i F(\alpha_i) = -N_f F(a) + (1-N_f) F(\Gamma) \text{.} 
\]

The total number of particles is constrained to take the form
\[
q = -N_f m + (1-g)
\]
for an integer \(m\). Incorporating the \(N_f^g\) states associated to the gauge field, which we again assumes enters in the simplest way, the total number of states for a given \(q\) is then
\[
N_f^g {N_f \mc{B} + (N_f-1)(1-g) \choose - N_f m + (1-g) } \text{.}
\]
This is again the vortex count.

This leads us to conjecture a low-temperature duality of nonrelativistic \((2+1)\)-dimensional theories of the schematic form
\[
U(1)_{N_f} + N_f \text{ fundamental scalars} \,\, \leftrightarrow \,\, U(1)^{N_f}/U(1) + N_f \text{ fermions } + \text{ BF couplings} \text{,} 
\]
taking vortices to fermions. 
This is highly reminiscent of the mirror symmetry duality
\[
U(1) + N_f \text{ hypermultiplets} \,\, \leftrightarrow \,\, U(1)^{N_f}/U(1) + N_f \text{ hypermultpliets}
\]
of three-dimensional gauge theories with \(\mc{N}=4\) supersymmetry \cite{intriligatorMS}. The appearance of a Chern--Simons term in the vortex theory can be thought of from this perspective as arising from integrating out the fermionic fields of the supersymmetric theory.

It would be interesting to understand more precisely the nature of three-dimensional mirror symmetry in this `quantum Hall regime' (a regime considered for three-dimensional Chern--Simons theories with \(\mc{N}=2\) supersymmetry in \cite{leblancES}).

\section{Conclusions}
\label{sec:conclusion}

By computing a fairly `soft' topological invariant, we have understood aspects of the behaviour of quantum vortices in nonrelativistic gauge theories. We have found that nonAbelian vortices with \(N_c = N_f = \lambda\) behave as fermions in a nonAbelian Landau level. More generally, we have found evidence that quantum vortices may be regarded as composite objects, built of locally-defined anyonic objects. 

Our results are consistent with duality: vortices in these theories are dual to the composite particles which constitute the correct low-temperature degrees of freedom in certain nonrelativistic electron fluids. It is plausible therefore that the general result \eqref{eq:generalvortexcount} contains information about rather general quantum Hall fluids. 


Our results, particularly at the special point \(\lambda = N_f\), live in the shadow of three-dimensional mirror symmetry (originally due to \cite{intriligatorMS}). As we alluded to above, it seems plausible that the duality we have discussed could be derived by a suitable deformation of conventional mirror dualities. Conversely, if such a deformation was understood, our results could be viewed as evidence for three-dimensional mirror symmetry. The use of exact results in topological quantum mechanics to probe and provide evidence for mirror symmetry has been of interest recently \cite{bullimoreTI,bullimoreTHS,crewFTI}.

As we commented briefly for the Abelian case, it is plausible that the vortex-hole (a)symmetry of the vortex quantum mechanics can be related to a nonrelativistic version of the bosonisation duality of \cite{jensenMB}. It would be interesting to develop this in detail, particularly for the nonAbelian case. As for the relation to mirror symmetry, this would require a good understanding of the deformation of three-dimensional gauge theories into the `quantum Hall regime'.

We have seen that the expression for the index is simplest when the Chern--Simons level \(\lambda\) is equal to the number of flavours \(N_f\). The general result \eqref{eq:generalvortexcount} contains more though. In particular, in the semiclassical limit \(\lambda \to \infty\), the index becomes the volume of the vortex moduli space (this can be seen by considering the large \(\lambda\) limit of \eqref{eq:HRR}). The volumes of the vortex moduli spaces describe the statistical mechanics of vortex gases \cite{mantonSMV} and so represent interesting quantities. They have been computed in \cite{mantonVVM,miyakeVMS,ohtaHCB,etoSMV}, although there remains some confusion in the case of nonAbelian vortices on surfaces of genus greater than zero. The results of this paper may be used to understand this issue.

Our results could also be turned in a rather different direction. In \cite{waltonGMS}, following ideas of \cite{bartonMS}, it was argued out that the geometric quantisation of the moduli space of vortices in a \(U(1)\) gauge theory with two fundamental flavours in the infrared limit of \(\mc{A} \to \infty\) describes the low energy quantum dynamics of magnetic Skyrmions, which are (approximately stable) solitons in ferromagnetic materials. The extent to which this is true is not clear, because it is a statement that is based on the micromagnetic approximation, which may not be valid at quantum scales. If it is a useful idea, the limit \(\mc{A} \to \infty\) introduces an infinite infrared degeneracy which can be regularised by putting the magnetic Skyrmions in a harmonic trap, as discussed in \autoref{subsec:harmonictrap}.

A special feature of \emph{nonrelativistic} gauge theories in \(2+1\) dimensions is that one can change the sign of the Higgs potential coupling without breaking the theory. This allows one to consider regimes where so-called \emph{exotic vortices}, such as Jackiw--Pi or Ambj{\o}rn--Olesen vortices, may exist \cite{jackiwSSS, ambjornAS, mantonFVE,turnerQO}. Some applications of the results of this paper to  exotic vortices will be given in forthcoming work.

\section*{Acknowledgements}

I am grateful to Nick Manton, Nick Dorey, and Bernd Schroers for helpful comments. This work has been partially supported by an EPSRC studentship and by STFC consolidated grants ST/P000681/1, ST/T000694/1.

\bibliographystyle{JHEP}
\bibliography{bibliography.bib}

\end{document}